\documentclass[sigconf]{acmart}

\usepackage{diagbox}
\usepackage{float}
\usepackage{booktabs}
\usepackage{enumitem}
\usepackage{multicol}
\usepackage[ruled,linesnumbered,vlined]{algorithm2e}

\AtBeginDocument{%
  }

\copyrightyear{2025} 
\acmYear{2025} 
\setcopyright{cc}
\setcctype{by}
\acmConference[KDD '25]{Proceedings of the 31st ACM SIGKDD Conference on Knowledge Discovery and Data Mining V.1}{August 3--7, 2025}{Toronto, ON, Canada}
\acmBooktitle{Proceedings of the 31st ACM SIGKDD Conference on Knowledge Discovery and Data Mining V.1 (KDD '25), August 3--7, 2025, Toronto, ON, Canada}
\acmDOI{10.1145/3690624.3709162}
\acmISBN{979-8-4007-1245-6/25/08}

\begin{document}
\title{DIPS: Optimal Dynamic Index for Poisson $\boldsymbol{\pi}$ps Sampling}

\author{Jinchao Huang}
\affiliation{%
  \institution{The Chinese University of Hong Kong}
  \city{Hong Kong}
  \country{China}}
\email{jchuang@se.cuhk.edu.hk}

\author{Sibo Wang}
\affiliation{%
  \institution{The Chinese University of Hong Kong}
  \city{Hong Kong}
  \country{China}}
\email{swang@se.cuhk.edu.hk}

\renewcommand{\shortauthors}{Huang et al.}




\begin{abstract}
This paper addresses the Poisson $\pi$ps sampling problem, a topic of significant academic interest in various domains and with practical data mining applications, such as influence maximization. The problem includes a set $\mathcal{S}$ of $n$ elements, where each element $v$ is assigned a weight $w(v)$ reflecting its importance. The goal is to generate a random subset $X$ of $\mathcal{S}$, where each element $v \in \mathcal{S}$ is included in $X$ independently with probability $\frac{c\cdot w(v)}{\sum_{v \in \mathcal{S}} w(v)}$, where $0<c\leq 1$ is a constant. The subsets must be independent across different queries. While the Poisson $\pi$ps sampling problem can be reduced to the well-studied subset sampling problem, updates in Poisson $\pi$ps sampling, such as adding a new element or removing an element,  would cause the probabilities of all $n$ elements to change in the corresponding subset sampling problem, making this approach impractical for dynamic scenarios. To address this, we propose a dynamic index specifically tailored for the Poisson $\pi$ps sampling problem, supporting optimal expected $\mathcal{O}(1)$ query time and $\mathcal{O}(1)$ index update time, with an optimal $\mathcal{O}(n)$ space cost. Our solution involves recursively partitioning the set by weights and ultimately using table lookup.  The core of our solution lies in addressing the challenges posed by weight explosion and correlations between elements. Empirical evaluations demonstrate that our approach achieves significant speedups in update time while maintaining consistently competitive query time compared to the subset-sampling-based methods.
\end{abstract}

\maketitle
\def\integer{\mathbb{Z}}
\def\real{\mathbb{R}}
\def\lc{\lceil}
\def\rc{\rceil}
\def\lf{\lfloor}
\def\rf{\rfloor}

\newcommand\mS[0]{\mathcal{S}}
\newcommand\mD[0]{\mathcal{D}}
\newcommand\mR[0]{\mathbb{R}}
\newcommand\mC[0]{\mathcal{C}}
\def\C{\mathcal{C}}
\def\D{\mathcal{D}}
\renewcommand{\S}{\mathcal{S}}

\renewcommand{\max}[2]{\mathrm{max}\{#1,#2\}}
\renewcommand{\min}[2]{\mathrm{min}\{#1,#2\}}
\newcommand\norm[1]{|#1|}
\newcommand{\set}[1]{\{#1\}}
\newcommand{\setht}[2]{\{#1,\ldots, #2\}}
\newcommand\normS[0]{|\mS|}
\newcommand{\logb}{\log_\mathrm{b}}

\newcommand{\bin}[2]{\mathit{Binomial}(#1,#2)}
\newcommand{\geo}[1]{\mathit{Geometric}(#1)}
\newcommand{\unif}[2]{\mathit{Uniform}(#1,#2)}
\newcommand\pv[0]{\textbf{p}(v)}
\def\bp{\textbf{p}}

\newcommand{\que}[0]{\emph{query}}
\newcommand\gen[0]{\emph{generate}}
\newcommand\gentwo[2]{\emph{generate}(#1,#2)}
\newcommand\upd[0]{\emph{change\_w}}
\newcommand\updtwo[2]{\emph{change\_w}(#1,#2)}
\newcommand\ins[0]{\emph{insert}}
\newcommand\instwo[2]{\emph{insert}(#1,#2)}
\newcommand\del[0]{\emph{delete}}
\newcommand\delone[1]{\emph{delete}(#1)}

\newcommand\compthree[3]{#1\leq#2\leq#3}
\newcommand\ifromoneton[2]{#1\in[#2]}

\newcommand\Bigo[1]{\mathcal{O}\left(#1\right)}
\newcommand\bigo[1]{\mathcal{O}(#1)}
\newcommand\bigone[0]{\bigo{1}}
\newcommand\bigolog[0]{\bigo{\log n}}
\newcommand\bigonlog[0]{\bigo{n\log n}}
\newcommand\bigon[0]{\bigo{n}}
\newcommand\bigonorms[0]{\bigo{|\mathcal{S}|}}
\newcommand\bigomu[0]{\bigo{1+\mu_\mS}}
\newcommand\bigologmu[0]{\bigo{\log n+\mu}}
\newcommand\bigologk[0]{\bigo{\log n+k}}

\newcommand\mul[0]{\mu_{S_l}}
\newcommand\ssp[0]{\textit{subset sampling problem}}

\SetKwInput{KwInput}{Input}
\SetKwInput{KwOutput}{Output}

\def\header{\vspace{1mm} \noindent}
\def\figcapdown{\vspace{-0mm}}

\newcommand{\revise}[1]{\begin{color}{blue}{#1}\end{color}}
\newcommand{\todo}[1]{\textcolor{red}{\bf [TO DO: #1]}}
\newcommand{\sibo}[1]{{\textcolor{blue}{[Sibo: #1]}}}

\newcommand\pps[0]{Poisson $\pi$ps sampling}
\newcommand\dub[0]{\text{DIPS}}
\newcommand{\rdss}{R-ODSS}
\newcommand{\rhss}{R-HSS}
\newcommand{\rbss}{R-BSS}
\newcommand{\dips}{DIPS}

\newtheorem{problem}{Problem}


\section{Introduction}   
In the field of statistical sampling, \emph{Poisson probability-proportional-to-size (Poisson $\pi$ps) sampling} is a technique recognized for its effectiveness in handling heterogeneous populations. This method assigns each population unit an inclusion probability proportional to an auxiliary variable indicative of the unit's size or significance. It ensures that larger or more critical units have a higher likelihood of being included in the sample, thereby enhancing the precision of resultant estimates. The flexibility of this technique, allowing for tailored inclusion probabilities based on prior knowledge or external data, makes it a powerful tool in practical applications. 

\pps~ is particularly advantageous in domains such as survey sampling~\cite{survey-sampling,survey-sampling2} and business statistics~\cite{business,sps,business3}, where population units exhibit significant variation in importance or size. Moreover, \pps~ is a pivotal cornerstone in important database and data mining applications. For instance, commercial databases offer efficient sampling methods for SQL and XQuery queries~\cite{Db2} and some queries need to sample each record based on its weight and the total weights of all records~\cite{weight1,weight2,weight3,weight4}. In the context of \emph{Influence Maximization} (IM), a fundamental data mining task, \pps~ is used to generate random RR-sets \cite{BorgsBCL14,Guo0WC20Subsim,GuoFZW23,FengCGZW24,GuoWWLT22,BianGWY20}. The generation of random RR-sets involves randomly selecting a vertex as the target and then performing a stochastic \emph{breadth-first search} (BFS) in the reverse direction of the edges. Here, ``stochastic'' implies that the existence of each incoming edge (and consequently, the in-neighbor) of a current visited vertex is determined in a probabilistic manner. In the widely-used \emph{Weighted Cascade} (WC) model \cite{Guo0WC20Subsim,TangTXY18OPIMC}, each incoming edge exists with a probability that is proportional to its weight. When no information about the edge $\langle u, v \rangle$ is given, the weight is typically assigned equally. Consequently, the probability of the edge $\langle u, v \rangle$ is given by $1/d_{in}(v)$, where $d_{in}(v)$ is the in-degree of vertex $v$. However, existing studies, e.g., \cite{GoyalBL10,KutzkovBBG13}, show that such a solution ignores the importance of different in-neighbors to vertex $v$ and setting different weights is suggested based on historical data. Given different weights, then the goal is to sample the set of incoming neighbors according to the \pps.

Despite its importance, \pps~ presents challenges, particularly in balancing the sampling efficiency and index maintenance. Current approach is to calculate the inclusion probability of each element and convert the \pps~ problem into the \emph{subset sampling} problem, where each element is associated with an independent inclusion probability instead of a weight. However, this method struggles with dynamic updates. In particular, suppose we have a set of $n$ elements with weights 1 to $n$ and then we add a new element with weight $n^3$. The probability of all elements get affected significantly and cannot be efficiently resolved with the subset sampling approach since this necessitates $n$ update operations in the corresponding subset sampling problem, rendering this approach prohibitively expensive for dynamic scenarios. However, in real-world scenarios, data is increasingly dynamic, creating a pressing need for efficient algorithms that can adapt to changes in real-time. 
 Thus, an algorithm specifically designed for \pps~ is needed to efficiently accommodate dynamic settings.

In this paper, we introduce \dips\footnote{Optimal \underline{D}ynamic \underline{I}ndex for \underline{P}oisson $\pi$ps \underline{S}ampling}, an optimal dynamic index maintenance solution designed to enhance the efficiency and scalability of \pps. Our framework achieves expected optimal $\mathcal{O}(1)$ query and index update times for both insertions and deletions, while maintaining a space cost of $\mathcal{O}(n)$, necessary for storing the $n$ elements.
The core idea of our DIPS framework involves recursively partitioning the input set by weights and integrating advanced techniques to manage the impact of a single weight update on the inclusion probabilities of all elements. This approach reduces the problem size to a manageable level, at which point a carefully designed lookup table for \pps~ is employed.
There are challenges specific to \pps~ that do not exist in subset sampling. One major challenge is weight explosion, which can occur when reducing the problem to smaller sizes. In \pps, element weights can be unbounded, whereas in subset sampling, the element inclusion probabilities are confined to the $[0,1]$ range. It is necessary to find a way to bound the weights so that a table lookup technique can be used effectively. Another challenge is that the inclusion probabilities of all elements in \pps~ are correlated, unlike in subset sampling where they are independent. This correlation complicates the design of a lookup table.
Our proposed DIPS framework effectively tackles these challenges specific to \pps, offering a solution for \pps~ with optimal time complexity in dynamic settings.

With its optimal time complexity, we further implement \dips~ and show that our \dips~ is also practically efficient. Experimental evaluations demonstrate the superb efficiency of our solution for both query processing and index update. We further integrate our solution into the well-known data mining application, the Influence Maximization problem, and show that our solution advances all previous methods under the dynamic setting.

\header 
{\bf Organization.} The remainder of this paper is structured as follows: Sec. \ref{sec:preliminaries} provides essential background and demonstrates why existing optimal solutions for subset sampling cannot be readily adapted to handle dynamic \pps. Sec. \ref{sec:our-solution} presents our main solution. Sec. \ref{sec:exp} presents experimental results demonstrating our solution's efficiency. Sec. \ref{sec:dynamic-im} shows how integrating \dips~ into dynamic Influence Maximization yields improved query and update efficiency. Sec. \ref{sec:conclusion} concludes the paper.


\section{Preliminaries}\label{sec:preliminaries}
    \subsection{Problem Formalization}\label{sec:problem-definition} 
            \header
        {\bf Poisson $\boldsymbol{\pi}$ps sampling.} 
        Consider a set $\S=\setht{v_1,v_2}{v_n}$ of elements, where each $v\in\S$ is assigned a weight $w(v)$. Define $W_\S=\sum_{u\in\mS}w(u)$. The \pps~query aims to randomly sample a subset $X$ of $\S$, such that each $v\in\S$ is independently sampled into $X$  with probability $\frac{c\cdot w(v)}{W_\S}$, where $0<c\leq 1$ is a constant. Importantly, the sampled subset must be independent of those returned by previous queries. More formally, we have the following definition.

        \begin{problem}[Poisson $\pi$ps Sampling (PPS)] {
            Given a set $\S$ of $n$ elements, a constant $c\in(0,1]$, and a function $w:\S\to\real_{\geq0}$, draw a random subset $X$ from $\S$ with
                $$
                \Pr[X=T\subseteq \S]=\left( \prod_{v\in T}\frac{c\cdot w(v)}{W_\S} \right)\left( \prod_{v\in {\S\setminus T}}(1-\frac{c\cdot w(v)}{W_\S}) \right).
                $$
        } \end{problem}
        We call the tuple $\Phi=\left<\S,w,c\right>$ a \pps~problem instance. Note that we require $c\in(0,1]$ here to ensure that the probabilities are well-defined.
        
        We further define three dynamic operations: {\em (i)} $\updtwo{v}{w}$, which resets the weight $w(v)$ to $w$; {\em (ii)} $\instwo{v}{w}$, which inserts $v$ into $\mathcal{S}$ and sets $w(v)$ to $w$; and {\em (iii)} $\delone{v}$, which deletes $v$ from $\mathcal{S}$. Without loss of generality, we assume that no duplicate keys exist in $\mathcal{S}$. These operations are collectively referred to as \emph{update} operations. It is important to note that a single \emph{update} operation can affect the sampling probabilities of all elements, even if it only changes the weight of one element. The main challenge, therefore, lies in supporting constant-time \emph{update} operations, despite the potential changes to all sampling probabilities.

        \header
        {\bf Subset sampling.}     
        A problem closely related to our studied \pps~is the subset sampling problem. In this problem, we are given the same set $\mS$ of element. However, each $v\in\S$ is assigned a probability $p(v)$ instead of a weight. The subset sampling query aims to randomly sample a subset $Y$ of $\mathcal{S}$, where each $v \in \S$ is independently sampled into $Y$ with probability $p(v)$. As with \pps, the sampled subset must be independent of those returned by previous queries. Another key difference is that the update operation in the subset sampling problem changes the sampling probability of an element rather than its weight. More formally, we have the following definition.
        \begin{problem}[Subset Sampling (SS)] {
            Given a set $\S$ of $n$ elements, and a function $p:\S\to[0,1]$, draw a random subset $Y$ from $\S$ with
                $$
                \Pr[Y=T\subseteq \S]=\left( \prod_{v\in T}\pv \right)\left( \prod_{v\in {\S\setminus T}}(1-\pv) \right).
                $$
        } \end{problem}

        We call the tuple $\phi=\left<\S,p\right>$ a subset sampling~problem instance. 
        Notice that for this problem, all the update operations act on $p(v)$.
        
        \header
        {\bf Remark.} Throughout this paper, we assume that generating a random value from $[0,1]$ or a random integer from $[a,b]$ requires constant time. Furthermore, for a truncated geometric random variable $G$ with parameters $p, q\in[0,1]$, where $\Pr[G = i]=p(1-p)^i/q$, we can generate samples in constant time~\cite{1+u_ss, taocpv2} using the formula $\left\lfloor\frac{\log(1 - q\cdot\unif{0}{1})}{\log(1 - p)}\right\rfloor$.

    \subsection{Related Work} \label{sec:related-work}
        In the literature, research work mostly related to the \pps~ problem is the line of research work on subset sampling. Indeed, if we only consider the static setting, there is almost no difference after converting the weights into probabilities. 
        The subset sampling problem has been studied for decades. Earlier research \cite{constant_p_ss} made progress in understanding special cases of the problem, such as the case where $\pv=p$ for all $v \in \S$. In this case, they achieve $\bigo{1+\mu_\S}$ expected query time, where $\mu_\S$ is the sum of the probabilities of all elements. The general case of the subset sampling query is considered in \cite{logn+u_ss}, where Tsai et al. design algorithms with $\bigon$ space and $\bigo{\log n + \mu_\S}$ query time. More recently, Bringmann and Panagiotou \cite{1+u_ss} design an algorithm that solves the subset sampling problem with $\bigon$ space and $\bigo{1+\mu_\S}$ expected query time. They additionally prove that their solution is optimal.  
        It is worth noting that all of the solutions mentioned above consider a static dataset that fits in memory. In real-world applications, however, the dataset may change over time. 
        To tackle this challenge, most recently, 
        two independent work~\cite{ODSS,huang2023subsetsamplingextensions} both present optimal solutions for the dynamic subset sampling problem with 
        $\bigo{1+\mu_\S}$ expected query time and $\bigo{1}$ index update time using linear space. However, if we consider the dynamic setting, the \pps~ is clearly more challenging than the dynamic subset sampling problem and the solution in \cite{ODSS} cannot tackle the this challenge. More details of their solution \cite{ODSS} and why their solution cannot be migrated to the \pps~ problem will be elaborated shortly. To our knowledge, there is no existing work on solving \pps~ under the dynamic setting. 

    \subsection{Relation Between PPS and SS} 
    Given a PPS problem instance $\Phi=\left<\S,w,c\right>$, where $\norm{\S}=n$, we can calculate $p_w(v)=\frac{c\cdot w(v)}{\sum_{u\in\mS}w(u)}$ for each $v\in\S$ in $\bigon$ time and get an SS problem instance $\phi=\left<\S,p_w\right>$.  Suppose that we can answer the SS query in $\bigo{Q(n)}$ time and perform an update on $\phi$ in $\bigo{M(n)}$ time, after $\bigo{I(n)}$ time initialization. Then the above reduction already gives us a solution for the PPS problem with $\bigo{Q(n)}$ query time after $\bigo{I(n)}$ initialization time. However, since an update operation on $\Phi$ would change $p_w(v)$ for all $v\in\S$, we need to invoke $\bigon$ update operation on $\phi$, which cost $\bigo{nM(n)}$ time. Indeed, if we need to do $\bigon$ update, we could reconstruct the whole data structure instead. However, the $\bigo{I(n)}$ cost for a single update to $\Phi$ is still too costly. Our goal is to finally achieve $\bigone$ update time while maintaining the same query time and initialization time asymptotically. Thus we need to specifically design an algorithm for PPS to achieve this goal.

    \subsection{SOTA Solution for SS}
    Among the literature exploring SS, ODSS~\cite{ODSS} is the state-of-the-art solution for the dynamic setting. Its key insight is twofold. For elements with probabilities less than $\frac{1}{n^2}$, the probability of sampling at least one of them is below $\frac{1}{n}$. This allows for a brute force approach in handling these rare cases. For elements with probabilities exceeding $\frac{1}{n^2}$, the solution partitions them into $O(\log n)$ groups. In group $i$, the probabilities of elements fall within the range $(2^{-i}, 2^{-i+1}]$. The benefit of this partition is that within each group, subset sampling can be performed efficiently. Consequently, the problem reduces to performing subset sampling on $O(\log n)$ groups. After two rounds of reduction, the problem size further decreases to $m = O(\log \log n)$. The solution then discretizes the interval $[0,1]$ into $O(\log \log n)$ equal scales and pre-computes a lookup table for SS instances where the problem size is $m$ and the probabilities correspond to these scales. For queries, each sampling probability is rounded up to the nearest scale, and the lookup table is consulted to obtain a set $X$. Finally, rejection sampling is applied to the elements in $X$ to correct for the overestimated probabilities.

    \subsection{Difficulty in Migrating ODSS to PPS}
    The solution presented in ODSS is designed for the subset sampling and cannot be easily adapted for use in \pps. 

    Firstly, in SS, if we group several elements together, the probability of sampling at least one element from that bucket naturally has an upper bound 1, which is given by the probability semantic. However, in PPS, if we put several elements into a bucket, we need to maintain the total weight of this bucket. Thus, the weight input grows as we group the elements together. The weight would finally explode after even constant rounds of recursion. Notice that we aim to finally opt for table lookup. To perform table lookup, the problem should have as few statuses as possible. Thus, it is also more challenging for \pps. 

    Secondly, in SS, we always only need to care about the same $t=\lceil 2\log n \rceil$ ranges $(\frac{1}{2},1], (\frac{1}{4}, \frac{1}{2}],\cdots,(2^{-t}, 2^{1-t}]$. In PPS, however, as we have no restriction on the input weight, the top $\bigo{\log n}$ ranges might be dynamically changing with every update. 

    Finally, in SS, the elements are totally unrelated to each other. Thus, we can sample from the buckets in hierarchy in a bottom-up fashion. In PPS, however, all the elements are related to each other, based on the expression of the sampling probability. Thus, we need to sample from the grouping hierarchy top down. Another challenge caused by the correlation between the elements is that, in the final table lookup step, probability or weight round up is required to reduce status number. We can safely round up the sampling probability of each element for SS. In PPS, however, if we simply round up each weight, we can not guarantee that the probability that each element being sampled is at least that of before rounding up, posing a challenge in the lookup table design. 


\section{A Dynamic Index for PPS} \label{sec:our-solution}
    We propose DIPS, a dynamic data structure for PPS that achieves optimal time and space complexity. At a high level, our solution consists of three key components:

    \header
    {\bf Building Blocks.} We first develop two fundamental data structures that can handle special cases efficiently - one for elements with near uniform weights, and another for elements with weights significantly lower than average.

    \header
    {\bf Size Reduction.} To handle arbitrary weights, we group elements into buckets based on their weight ranges, then organize these buckets into chunks. By carefully managing the significant chunks (containing the heaviest elements) and applying our building blocks, we reduce the problem size from $n$ to $\bigo{\log n}$ elements.

    \header
    {\bf Table Lookup.} After reduce the problem size twice to $\bigo{\log\log n}$, we develop a lookup table approach that discretizes weights while preserving sampling probabilities through rejection sampling. 

    The combination of these techniques yields a data structure that requires $\bigon$ space and supports queries in $\bigone$ expected time, with $\bigone$ expected time for updates, insertions, and deletions. We elaborate on each component in the following subsections.
    
    \subsection{Building Blocks}\label{sec:blocks} 
        We first present two fundamental data structures that handle special cases of PPS efficiently. 
        
        \begin{lemma}[bounded weight ratio]\label{wss_grs} {
            Given a PPS problem instance $\Phi=\left<\S, w, c\right>$ and a subset $T\subseteq\S$, if all weights in $T$ fall within a bounded ratio range $w(T)\subseteq(\bar{w}/b, \bar{w}]$ where $b$ is a constant and $\bar{w}\in\real_{>0}$, then we can:
                (i) Initialize a data structure $\D$ of size $\bigo{t}$ in $\bigo{t}$ time, where $t=|T|$;
                (ii) Support queries in $\bigo{1}$ expected time;
                (iii) Support updates, insertions and deletions in $\bigone$ expected time.
        } \end{lemma}

        The key insight is that when weights are bounded within a constant ratio, we can use rejection sampling with geometric random variables to efficiently generate samples. By maintaining a dynamic array and carefully choosing sampling parameters, we achieve constant expected time for all operations. The complete pseudocode and proof are provided in Appendix~\ref{apd:grs}.

        \begin{lemma} [subcritical-weight] \label{wss_naive} 
            Given a PPS problem instance $\Phi=\left<\S, w, c\right>$ and a subset $T\subseteq\S$, if all weights in $T$ are bounded above by $\bar{w}=\bigo{W_\S/n^2}$, then we can:
                (i) Initialize a data structure $\D$ of size $\bigo{t}$ in $\bigo{t}$ time, where $t=|T|$;
                (ii) Support queries in $\bigone$ expected time;
                (iii) Support updates, insertions and deletions in $\bigone$ expected time.
        \end{lemma}
        \begin{proof}
            The data structure and operations are identical to those in Lemma~\ref{wss_grs}. In this special case, since the weight of $T$ is subcritical to the total weight, i.e., $\frac{W_T}{W_\S}=\bigo{\frac{1}{n}}$, we can even scan the entire $T$ and toss a coin for each element to decide whether to include it in the sample. The expected query time is $O(1)+\frac{W_T}{W_\S}\cdot\bigon=\bigone$.
        \end{proof}        

    \subsection{Size Reduction}\label{sec:sr} 
        According to Lemma~\ref{wss_grs}, we can achieve the desired time and space complexity for PPS problem instances where weights fall within a bounded interval. For general instances with unrestricted weights, we can partition elements in $\S$ into disjoint buckets based on weight ranges. While sampling within individual buckets is straightforward, the challenge lies in sampling across buckets, as there could be $\bigon$ buckets in the worst case. A key observation is that the top $\bigo{\log n}$ buckets contain most of the total weight in $\S$. By applying Lemma~\ref{wss_grs} to these top buckets and Lemma~\ref{wss_naive} to the remaining elements (whose weights are bounded by $\bigo{W_\S/n^2}$), we effectively reduce the problem size to $\bigo{\log n}$. The main technical challenges are managing the exponential growth of bucket weights and handling dynamic changes in the top $\bigo{\log n}$ buckets.
        
        Let us first see some definitions before giving formal description of our size reduction method.

    \header{\bf Buckets and chunks.} Given the input set of elements, we regard these $n$ elements as ``zeroth-level'' elements. 
        For a zeroth-level element, if its weight falls within the interval $(b^{j},b^{j+1}]$, where $j$ is an integer and $b\geq 2$ is a fixed integer, we place it into bucket $B_j$. We say that a bucket $B_j$ is non-empty if it contains at least one element.
        We denote by $w(B_j)$ the total weight of elements in the bucket $B_j$, i.e., $w(B_j)=W_{\S\cap B_j}=\sum_{v\in \S\cap B_j}w(v)$.
        For a bucket $B$ that is empty, $w(B)$ is trivially defined as 0.
        Next, for a given bucket $B_j$, we say that it is in the chunk $C_t$ if and only if $j\in[t\lc\logb n\rc,\ldots,(t+1)\lc\logb n\rc-1]$. This is essentially assigning a slowly-changing boundary for each bucket according to its index\textemdash the bucket-chunk relation changes only when $n$ doubles or halves.
        We say that a chunk is non-empty if and only if it has at least one non-empty bucket. 
        We give an example in Appendix~\ref{apd:example} to help understand the concept of buckets and chunks.
        For a given non-empty chunk $C_t$ and a bucket $B_j$ in $C_{t}$, since we know that $j\in[t\lc\logb n\rc,\ldots,(t+1)\lc\logb n\rc-1]$, then we have that:
        $$1\cdot b^{t\lc\logb n\rc}< w(B_j)\leq n\cdot b^{(t+1)\lc\logb n\rc}.$$

        To explain, each non-empty chunk includes at least one non-empty bucket with at least one element whose weight is larger than $b^{t\lc\logb n\rc}$; each non-empty chunk includes at most $n$ elements and the weight of each element is at most $b^{(t+1)\lc\logb n\rc}$.
        
        Thus we can further normalize the weights of all buckets in a non-empty chunk $C_t$ by dividing them by $b^{t\lc\logb n\rc}$. After that, the weight of each bucket $B_j$ is normalized to be $w^\prime(B_j)=w(B_j)\cdot b^{-t\lc\logb n\rc}$. We now have:
        $$1< w^\prime(B_j)\leq n\cdot b^{\lc\logb n\rc}\leq bn^2.$$
        We denote the set of bucket IDs in $C_t$ as $\S_{C_t}$. For each $i\in\S_{C_t}$, we define $w_{C_t}(i)=w^\prime(B_i)$. Then we can define a PPS problem induced by $C_t$ as $\Phi_{C_t}=\left<\S_{C_t},w_{C_t},1\right>$, where each element corresponds to the ID of a bucket in $C_t$. When the context is clear, we will simply use $\Phi_C$ to denote the PPS problem instance defined on chunk $C$.
        The weight of a chunk $C_t$ before (resp., after) normalization is denoted as $w(C_t)$ (resp., $w^\prime(C_t))$, which is the total weights of buckets in $C_t$ before (resp. after) normalization, i.e., $w(C_t)=\sum_{i\in \S_{C_t}} w(B_i)$ (resp. $w^\prime(C_t)=\sum_{i\in \S_{C_t}} w'(B_i)$). 
                
        Let $C_r$ be the nonempty chunk with the highest index, such that for all other nonempty chunks $C_i$, we have $i < r$.  We call $C_r$, $C_{r-1}$ and $C_{r-2}$ \emph{significant chunks}. We denote the union of all significant chunks as $C^*=C_r\cup C_{r-1}\cup C_{r-2}$. The other chunks are called \emph{non-significant chunks}. 
        We claim that the weight of an element in the non-significant chunks is at most $\bigo{\frac{1}{n^2}}$ of the total weight of all the elements. This is because:
        $$
        \frac{b^{(r-3+1)\lc\logb n\rc-1}}{b^{r\lc\logb n\rc}}\leq b^{-2\logb n-1}=\frac{1}{bn^2}.
        $$

        Having defined these terms, we now prove the following lemma.
        \begin{lemma}[size reduction] \label{wss_sr}
        Given a \pps~problem instance $\Phi=\left<\S, w, c\right>$, 
            assume that using method $\mathcal{M}$, for each non-empty chunk $C$, after linear initialization time, we can create a data structure $\widetilde{\D}(C)$ of linear size for $\Phi_C$ with $\bigo{1}$ $\que$ time and $\bigone$ expected $\upd$ time. Then for the original problem $\Phi$, after $\bigon$ initialization time, we can create a data structure $\mD$ of size $\bigo{n}$ with $\bigo{1}$ expected $\que$ time and $\bigone$ expected $\upd$, $\ins$, and $\del$ time.
        \end{lemma}

        \begin{algorithm}[t]
        \caption{Init and Query operation of Lemma~\ref{wss_sr}}
        \label{alg:sr_init_q}
        \DontPrintSemicolon
        
        \SetKwProg{Init}{init$(\S, w)$}{:}{}
        \Init{} {
            $OldSize \gets |\S|$; $Size \gets |\S|$; $W_\S \gets 0$\;
            Initialize empty buckets $B_j$ for all possible $j$\;
            \ForEach{$v \in \S$}{
                $j \gets \lfloor \log_b w(v) \rfloor$\;
                Add $v$ to bucket $B_j$\;
                $w(B_j) \gets w(B_j) + w(v)$; $W_\S \gets W_\S + w(v)$\;
            }
            \ForEach{non-empty bucket $B_j$}{
                $t \gets \lfloor j/\lceil\log_b n\rceil \rfloor$\;
                Add $B_j$ to chunk $C_t$\;
                $w(C_t) \gets w(C_t) + w(B_j)$\; $w'(B_j) \gets w(B_j) \cdot b^{-t\lceil\log_b n\rceil}$\;
                $w'(C_t) \gets w'(C_t) + w'(B_j)$\;
                Create $D(B_j)$ using Lemma \ref{wss_grs}\;
            }
            \ForEach{non-empty chunk $C_t$}{
                Create $\widetilde{D}(C_t)$ using method $\mathcal{M}$\;
            }
        }

        \SetKwProg{Que}{query$()$}{:}{}
        \Que{} {
            $X \gets \emptyset$; $r \gets \lfloor \log_b W_\S/\lceil\log_b n\rceil \rfloor$\;
            \For{$i \gets r-2$ \KwTo $r$}{
                \If{$C_i$ is non-empty}{
                    $Y_i \gets$ query $\widetilde{D}(C_i)$\;
                    \ForEach{$j \in Y_i$}{
                        $Z \gets$ query $D(B_j)$\;
                        \ForEach{$v \in Z$}{
                            $u \gets \unif{0}{1}$\;
                            \If{$u < \frac{w(C_i)}{W_\S}$}{
                                $X \gets X \cup \{v\}$\;
                            }
                        }
                    }
                }
            }
            Sample from remaining chunks using Lemma \ref{wss_naive}\;
            \Return{$X$}\;
        }
        \end{algorithm}

        \begin{proof}
            The data structure, along with the query and update algorithms, are explained individually as follows.
        
            \header{\bf Structure.}
                We first scan $\S$ to put each element into its corresponding bucket. We then scan again to put each bucket into its corresponding chunk.
                We utilize method $\mathcal{M}$ to construct a structure $\widetilde{D}(C)$ for each non-empty chunk $C$.
                Additionally, we employ Lemma~\ref{wss_grs} to build a structure $D(B)$ for each bucket $B$. We also store several pieces of information as part of the data structure. This includes 
                the weights $w(C_t)$ and $w^\prime(C_t)$ for each non-empty chunk $C_t$, and the total weight $W_\S$. We also keep track of the initial size of $\S$ when initializing the data structure as $OldSize$, and the current size of $\S$ as $Size$. Given that the structures for each subproblem require linear space and initialization time, the total space and initialization time for the entire problem is $\bigo{n}$.
            
            \header{\bf Query.}
                To answer a query, we first find the largest ID of non-empty chunk $r$. Let $x=\frac{\lc\logb W\rc}{\lc\logb n\rc}$. It is easy to verify that the following should hold: $x-2\leq r\leq x$; otherwise, it is impossible for the total weight of the $n$ elements to be $W$. Then, we can easily check if $C_x$, $C_{x-1}$, $C_{x-2}$ is empty in order to identify the largest ID of non-empty chunk.
                Next, for each $i\in\{r-2,r-1,r\}$,
                we sample a subset $Y_i$ from the structure $\widetilde{D}(C_i)$, which takes $\bigone$ time according to the property of $\mathcal{M}$. 
                We further utilize the structure $D(B_j)$ to sample elements from bucket $B_j$ for each ID $j$ in $Y_i$. For each sampled element, we draw a random number $u$ from the uniform distribution $\unif{0}{1}$, if $u<\frac{w(C_i)}{W_\S}$, we finally add this element to the set $X$ to be returned. This step also takes $\bigone$ time according to Lem~\ref{wss_grs}. 
                Next, for the elements in the remaining chunks, since their weights satisfy the condition of Lemma \ref{wss_naive}, we can apply Lemma \ref{wss_naive} to sample a subset from it, which also takes $\bigone$ time according to Lem~\ref{wss_naive}. 
                
                Under this sampling scheme, each element $v\in B^v\subseteq\C_i$ for $r-2\leq i\leq r$ is sampled into $X$ with probability 
                \begin{align*}
                    &\frac{w(C_i)}{W_\S}\cdot\frac{w^{\prime}(B^v)}{w^{\prime}(C_i)}\cdot\frac{c\cdot w(v)}{w(B^v)}=\frac{w(C_i)}{W_\S}\cdot\frac{w(B^v)}{w(C_i)}\cdot\frac{c\cdot w(v)}{w(B^v)}=\frac{c\cdot w(v)}{W_\S}.
                \end{align*}
                Each element $v\in\S\setminus C_{sig}$ is sampled into $X$ with probability $\frac{c\cdot w(v)}{W_\S}$ according to Lemma~\ref{wss_naive}.
                The expected time complexity is: $\sum_{i=r-2}^{r}\bigo{1}+\bigo{1}=\bigo{1}.$

        \header{\bf Update.}
                Please refer to Appendix~\ref{apd:sr_update} for detailed description of the update operations.
        \end{proof}

    \subsection{Table Lookup and Final Result}\label{sec:table} 
        After applying size reduction twice, the problem reduces to solve a PPS instance $\Phi^o= \left< \S^o, w^o \right>$, where $\norm{\S^o}=\bigo{\logb \logb n}$ and $w^o(v)\leq b\lc\logb n\rc^2$ for all $v\in\S$.
        In this subsection, we propose a table lookup solution for the \pps~problem where the weights have an upper bound that grows exponentially with the problem size.

        \begin{lemma}[table lookup]\label{wss_lookup}
        Suppose $\S^o=\setht{1,2}{m}$ and for all $v\in\S^o$, $w^o(v)$ is greater than $1$ and is bounded by $b^{dm}$, where $d$ is a constant. Then for the \pps~problem instance $\Phi^o= \left< \S^o, w^o \right>$, we can maintain a data structure $\mD^o$ of size $m^{\bigo{m^2}}$ with $\bigone$ expected $\que$ time. Furthermore, $\mD^o$ can be maintained with $\bigone$ $\upd$ time.
        \end{lemma}
        If $m < b$, the above lemma holds trivially. For the remainder of the proof, we will consider the case where $m \geq b$.
        
        To build a lookup table that can correctly answer the query and has small size, we first discretize the weights while ensuring that the sampling probabilities are not scaled down. In the setting of PPS, this is not easily achievable. We explain with the following concrete example.
        \begin{example}
        Suppose $\S^o=\{1,2,3,4\}$, and the weight function is $w(1)=2.9$, $w(2)=7$, $w(3)=3.1$, $w(4)=4.7$. Then we have $\lc w(1)\rc=3$, $\lc w(2)\rc=7$, $\lc w(3)\rc=4$, $\lc w(4)\rc=5$. If we directly sample based on these scaled weights, the probability of $1$ being sampled would become $\frac{3}{3+7+4+5}=\frac{3}{19}$, which is smaller than $\frac{2.9}{2.9+7+3.1+4.7}=\frac{2.9}{17.7}$. Thus we can not directly sample based on the rounded-up weights.
        \end{example}
        To mitigate this problem, we maintain a lookup table with overestimated probabilities, where each element $v$ is sampled with probability $\frac{\lc w^o(v)\rc}{\sum_{v\in\S^o}\lc w^o(v)\rc-m}$ and then apply rejection sampling to correct for the overestimation. Notice that since $m\geq b\geq2$ and for each element $v$ we have $w(v)>1$, $\frac{\lc w^o(v)\rc}{\sum_{v\in\S^o}\lc w^o(v)\rc-m}\leq 1$ is guaranteed.
        
        Next, we are ready to give the detailed proof.
        \begin{proof} 
        The data structure, query operation, and update operation are as follows.
        
        \header{\bf Structure.}
        We define $\bar w(v)=\lc w^o(v)\rc$ for each $v\in\S^o$.
        We use a numeral system with radix $r\triangleq b^{dm}$ to encode $\bar w$, more specifically, $\bar w$ is encoded as:
        $$
        \lambda = \left( \bar w(m)  \bar w(m-1) \ldots \bar w(1) \right)_{r}.
        $$
        Then there are at most $r^m$ possible status of $\bar w$ as $\S^o$ undergoes \upd, i.e., each number in the numeral system with radix $r$ indicates a status of $\bar w$. We also maintain $W=\sum_{i=1}^m w(i)$ and $\overline{W} = \sum_{i=1}^m \lc w(i)\rc$.
        For a given $\bar w$ and a given subset $T$, we want the probability of sampling $T$ from the table to be 
        \begin{align*}
            \bar p(T)=\prod_{v\in T}\frac{\bar w(v)}{\overline{W}-m}\prod_{u\notin T}\frac{\overline{W}-m-\bar w(u)}{\overline{W}-m}\\=\frac{\prod_{v\in T}\bar w(v)\prod_{u\notin T}(\overline{W}-m-\bar w(u))}{(\overline{W}-m)^m}.
        \end{align*}
        
        Notice that for all subsets $T$ of $\S^o$, their probability sum up to $1$, i.e., $\sum_{T\subseteq \S^o} \bar p(T)=1$. Thus, we can create an array $A_\lambda$ of size $(\overline{W}-m)^m$, where each entry in the array can store a subset. Then we fill $(\overline{W}-m)^m \bar p(T)$ entries of $A_\lambda$ with $T$ for each subset $T$ of $\S^o$. To understand the table, we explain with the following example.
        \begin{example}
            Suppose $\S^o=\{1,2\}$, $b=2$, $d=2$, let us see what should be pre-computed for a specific status of $\bar {w}$, say $\bar w(1)=4,\bar w(2)=3$. First of all, the radix $r=2^{2\times2}=16$, then the $\bar {w}$ is encoded as $\lambda = (3\quad4)_{16}=3\times 16+4\times 1 =52$. We also have $\overline{W}=4+3=7$. Next we create an array $A_{52}$ of size $(\overline{W}-m)^m=(7-2)^2=25$. For subset $\emptyset\subset\S^o$, it is sampled with probability $\frac{7-2-4}{7-2}\times\frac{7-2-3}{7-2}=\frac{2}{25}$, so we fill the first $2$ entry of $A_{52}$ with $\emptyset$. For subset $\{1\}\subset\S^o$, it is sampled with probability $\frac{4}{7-2}\times\frac{7-2-3}{7-2}=\frac{8}{25}$, so we fill the next $8$  entry of $A_{52}$ with $\{1\}$. For subset $\{2\}\subset\S^o$, it is sampled with probability $\frac{7-2-4}{7-2}\times\frac{3}{7-2}=\frac{3}{25}$, so we fill the next $3$ entry of $A_{52}$ with $\{2\}$. For subset $\{1,2\}\subset\S^o$, it is sampled with probability $\frac{4}{7-2}\times\frac{3}{7-2}=\frac{12}{25}$, so we fill the next $12$ entry of $A_{52}$ with $\{1,2\}$. The case when $\bar{w}(1)$ and $\bar{w}(2)$ have other values can be handled by creating new arrays with similarly ways. 
        \end{example}
        As for each $\bar w$, we create an array $A_\lambda$ following the above process. The final space usage is bounded by $r^m\cdot(rm-m)^m\cdot m=m^{\bigo{m^2}}$.
        
        \header{\bf Query.}
        For query, we will look into the array $A_\lambda$ created for the $\bar w$ encoded by $\lambda$. We first uniformly sample an entry from $A_\lambda$, yielding a subset $T$. 
        Because $w(i)>1$, the expectation of the size of the $T$ is 
        $$
        \frac{\sum_{i=1}^m \bar w(i)}{\overline{W}-m}=\frac{\overline{W}}{\overline{W}-m}\leq\frac{2m}{2m-m}=2=\bigo{1}.
        $$
        For each element $v\in T$, we need to apply rejection sampling to correct for the overestimated sampling probability of $\frac{ \bar w(v)}{\overline{W}-m}$. Specifically, for each element $j\in T$, we draw a random number $u$ from $\unif{0}{1}$ and retain $j$ in the final sample if and only if $u<\frac{c\cdot w(v)}{\bar w(v)}\cdot\frac{\overline{W}-m}{W}$. In this way, each element is sampled with probability $\frac{ \bar w(v)}{\overline{W}-m}\cdot\frac{c\cdot w(v)}{\bar w(v)}\cdot\frac{\overline{W}-m}{W}=\frac{c\cdot w(v)}{W}$. Since the expectation of the size of $T$ is $\bigone$, and the rejection step cost $\bigone$ time, the query time is $\bigone$.

        \header{\bf Update.}
        For $\updtwo{v}{w}$ operation, it is equivalent to updating the $v$-th digit of the $r$-based number $\lambda$, which is a generalized \emph{bit operation}. Specifically, we update $\lambda$ to $$\lfloor\frac{\lambda}{r^v}\rfloor r^v+\lc w\rc r^{v-1}+\lambda\%r^{v-1},$$ which can be done in $\bigone$ time.
        \end{proof}

        \begin{algorithm}[t]
        \caption{Table Lookup}
        \label{alg:lookup}
        \DontPrintSemicolon
        
        \SetKwProg{Init}{init$(\S^o, w^o)$}{:}{}
        \Init{} {
            $r \gets b^{dm}$ \tcp*{Set radix for numeral system}
            $W \gets 0$; $\overline{W} \gets 0$; $\lambda \gets 0$\;
            \ForEach{$v \in \S^o$}{
                $\bar{w}(v) \gets \lceil w^o(v) \rceil$;
                $W \gets W + w^o(v)$;
                $\overline{W} \gets \overline{W} + \bar{w}(v)$\;
            }
            \For{$i \gets 1$ \KwTo $m$}{
                $\lambda \gets \lambda \cdot r + \bar{w}(i)$ \tcp*{Encode weights in base-r}
            }
            Create array $A_\lambda$ of size $(\overline{W}-m)^m$\;
            \ForEach{$T \subseteq \S^o$}{
                $p \gets \prod_{v\in T}\bar{w}(v)\prod_{u\notin T}(\overline{W}-m-\bar{w}(u))$\;
                Fill next $p$ entries of $A_\lambda$ with $T$\;
            }
        }

        \SetKwProg{Cw}{change\_w$(v, w_{new})$}{:}{}
        \Cw{} {
            $W \gets W - w^o(v) + w_{new}$;
            $\overline{W} \gets \overline{W} - \bar{w}(v) + \lceil w_{new} \rceil$\;
            $old\_digit \gets \bar{w}(v)$;
            $new\_digit \gets \lceil w_{new} \rceil$\;
            $\bar{w}(v) \gets new\_digit$\;
            \tcp{Update encoded number using bit operations}
            $\lambda \gets \lfloor\lambda/r^v\rfloor r^v + new\_digit \cdot r^{v-1} + \lambda\bmod r^{v-1}$\;
            $w^o(v) \gets w_{new}$\;
        }
        
        \SetKwProg{Que}{query$()$}{:}{}
        \Que{} {
            $X \gets \emptyset$\;
            $i \gets \unif{0}{(\overline{W}-m)^m-1}$\;
            $T \gets A_\lambda[i]$ \tcp*{Sample subset from lookup table}
            \ForEach{$v \in T$}{
                \If{$\unif{0}{1} < \frac{c\cdot w^o(v)}{\bar{w}(v)} \cdot \frac{\overline{W}-m}{W}$}{
                    $X \gets X \cup \{v\}$\;
                }
            }
            \Return{$X$}\;
        }
    \end{algorithm}

            After using Lemma \ref{wss_sr} twice, we only need to handle the PPS problem instance where the number of elements is $m=\lc\logb\lc\logb n\rc\rc$ and the weight function has a lower bound 1 and an upper bound $b(\lc\logb n\rc)^2\leq b(b^m)^2=b^{2m+1}\leq b^{3m}$. This satisfies the condition of Lem. ~\ref{wss_lookup}. Thus we can use the lookup table to end the recursion. the size of the table is 
            $$
            m^{\bigo{m^2}}=\bigo{b^{(\lc\logb\lc\logb n\rc\rc^2)(\logb\lc\logb\lc\logb n\rc\rc)}}=\bigo{b^{\logb n}}=\bigon.
            $$
            Therefore, we finally achieve at the following theorem by applying Lem.~\ref{wss_sr} twice and finally use Lem.~\ref{wss_lookup} to end the recursion.
            \begin{theorem}\label{sampling structure for WSS}
            Given a \emph{\pps} problem instance $\Phi=\left<\S, w, c\right>$, we can maintain a data structure $\mD$ of size $\bigon$ with $\bigo{1}$ $\que$ time. Furthermore, $\mD$ can be maintained with $\bigone$ expected $\upd$, $\ins$ and $\del$ time.
            \end{theorem}


\section{Experiments} \label{sec:exp}
We provide an experimental evaluation to compare the performance of our proposed DIPS with alternative approaches. All experiments are conducted on a Linux machine  with an Intel Xeon 2.3GHz CPU and 755GB of memory. The source code~\cite{code} for all methods is implemented in C++ and compiled with full optimization.

	\subsection{Experimental Settings}

        \header{\bf Competitors.}
        We evaluate our \dub~ against existing subset sampling methods (by reducing \pps~ to subset sampling): HeterogeneousSS~\cite{logn+u_ss} (referred to as \rhss), BringmannSS~\cite{1+u_ss} (referred to as \rbss), and ODSS~\cite{ODSS} (referred to as \rdss). 

        \header{\bf The distribution of weights.}
        We selected four types of distributions to represent the weights of the elements: the \emph{exponential distribution}, the \emph{normal distribution}, the \emph{half-normal distribution}, and the \emph{log-normal distribution}. The exponential distribution is a continuous probability distribution characterized by the \emph{probability density function} (PDF) $f(x)=\lambda e^{-\lambda x}$, where $\lambda$ is the rate parameter. The normal distribution, also known as the Gaussian distribution, is defined by the PDF $f(x) = \frac{1}{\sqrt{2\pi\sigma^2}} exp({-\frac{(x-\mu)^2}{2\sigma^2}})$, where $\mu$ is the mean and $\sigma^2$ is the variance. The half-normal distribution is derived from the absolute value of a zero-mean normal distribution. Its PDF is given by $f(x) = \frac{\sqrt{2}}{\sigma \sqrt{\pi}} exp({-\frac{x^2}{2\sigma^2}})$. The log-normal distribution describes a random variable whose logarithm is normally distributed. Its PDF is expressed as $f(x) = \frac{1}{x\sigma\sqrt{2\pi}} exp({-\frac{(\ln x - \mu)^2}{2\sigma^2}})$. 

        \header{\bf Parameter settings.}
        We set $c=1$ by default. For our method, \dub, we configured the parameter $b=4$. Following previous study~\cite{ODSS}, for the exponential distribution, we set $\lambda=1$; for the normal distribution, we used $\mu=0$ and $\sigma=\sqrt{10}$; the same value of $\sigma=\sqrt{10}$ was used for the half-normal distribution; for the log-normal distribution, we set $\mu=0$ and $\sigma=\sqrt{\ln 2}$. 

	\subsection{Correctness of Queries}
	To empirically verify the correctness of \dub~and its competitors, we conducted a series of experiments. According to the law of large numbers, the empirical probability of each element, derived from a substantial number of queries, should converge to its true probability, with accuracy improving as the number of queries increases. Following \cite{ODSS}, in our experiments, we calculate the empirical probability $\hat{p}(e)$ for each element $e$ and report the \emph{maximum absolute error}, defined as $\mathrm{max}_{e \in \mathcal{S}} |\hat{p}(e) - p(e)|$. Specifically, we set the number of elements to $n = 10^5$ and perform 1,000 update operations, consisting of 500 insertions followed by 500 deletions. Subsequently, we conduct \pps~ queries for $\{10^3, 10^4, 10^5, 10^6, 10^7\}$ iterations. This approach allows us to test both the update and query algorithms. Figure \ref{fig:err} illustrates the maximum absolute error across all methods for varying numbers of repetitions. Our results indicate that the maximum absolute error decreases as the number of repetitions increases and all methods follow identical trends, supporting the correctness of these methods.
	\begin{figure*}[t]
		\centering
		\begin{small}
			\begin{tabular}{cccc}
					\multicolumn{4}{c}{\hspace{-6mm} \includegraphics[height=3mm]{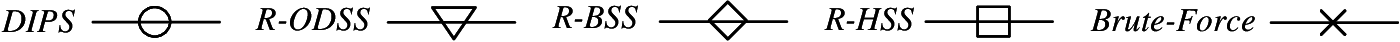}}  \\[0mm]
				\includegraphics[height=26mm]{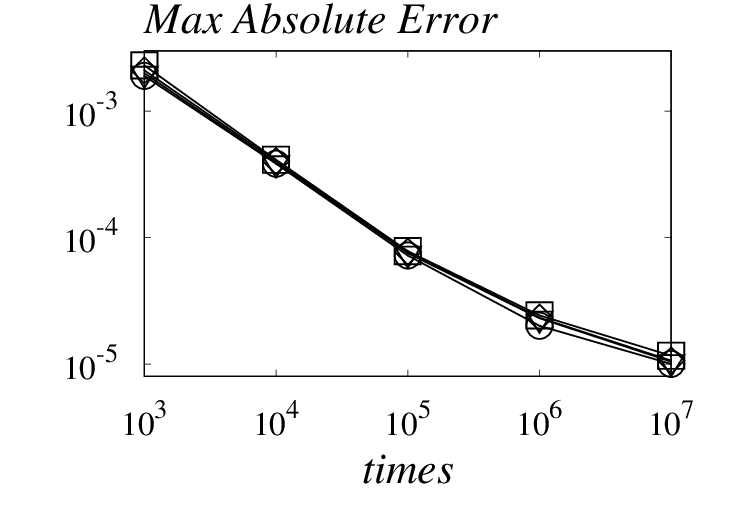} &
				\hspace{-2mm} \includegraphics[height=26mm]{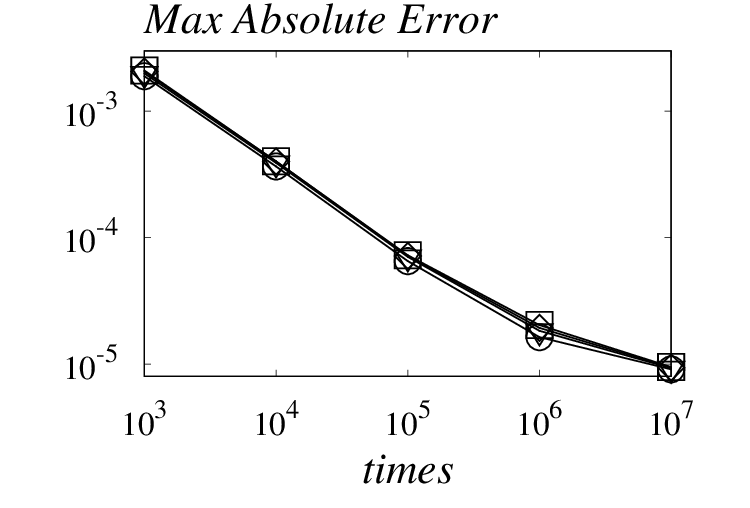} &
				\hspace{-2mm} \includegraphics[height=26mm]{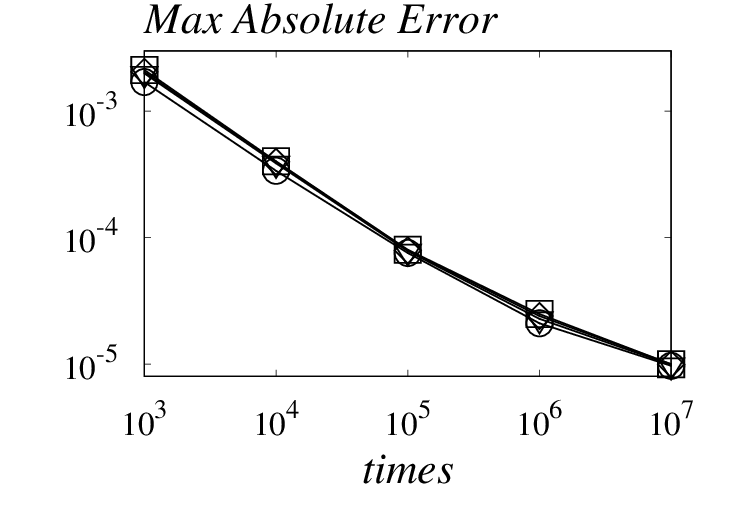}&
				\hspace{-2mm} \includegraphics[height=26mm]{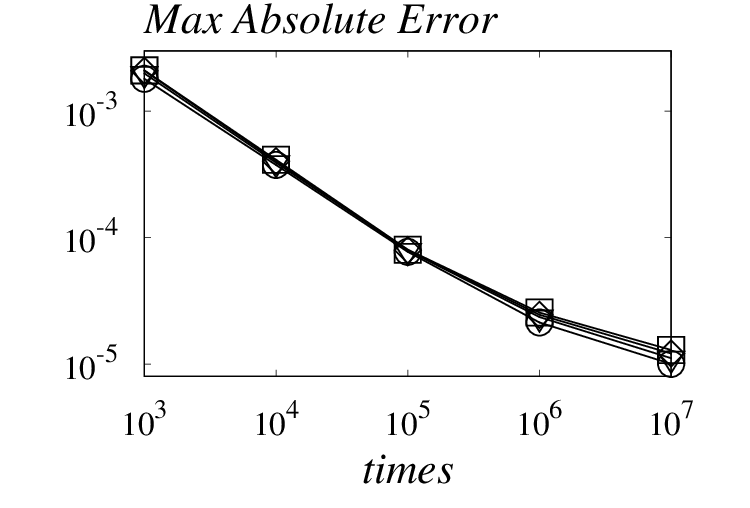}
				\\[-2mm]
				\hspace{-1mm} (a) Exponential distribution &
				\hspace{-1mm} (b) Normal distribution &
				\hspace{-1mm} (c) Half-normal distribution&
				\hspace{-1mm} (d) Log-normal distribution \\[-1mm]
			\end{tabular}
			\vspace{-2mm}
			\caption{Max absolute error v.s. repeat times on different distributions (in seconds). ($n=10^5$)} \label{fig:err}
			\vspace{-2mm}
		\end{small}
	\end{figure*}

	\begin{figure*}[t]
		\centering
		\begin{small}
			\begin{tabular}{cccc}
					\multicolumn{4}{c}{\hspace{-6mm} \includegraphics[height=2.3mm]{figure/qvsu_normal_legend.eps}}  \\[-2mm]
				\hspace{-2mm} \includegraphics[height=26mm]{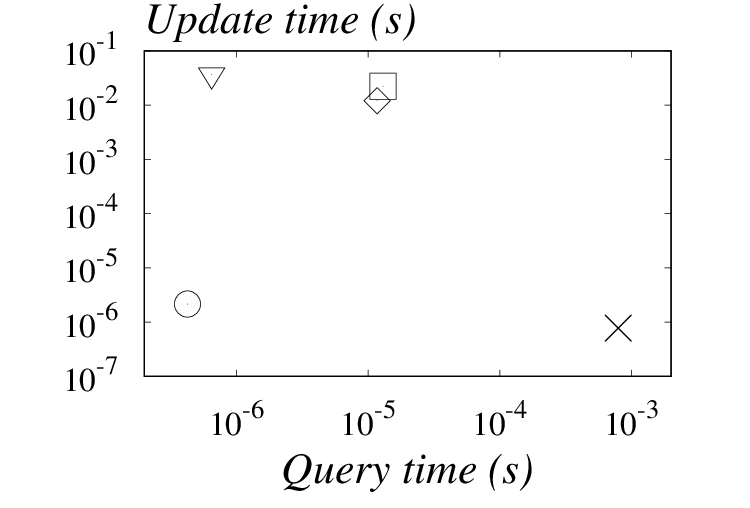} &
				\hspace{-2mm} \includegraphics[height=26mm]{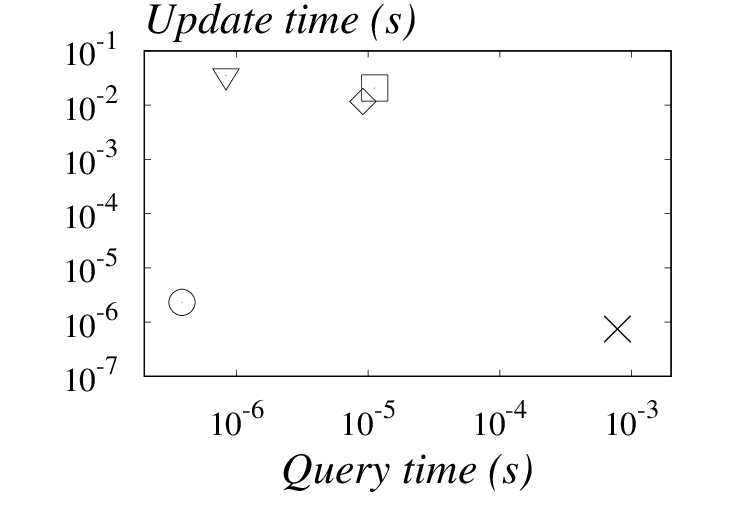} &
				\hspace{-2mm} \includegraphics[height=26mm]{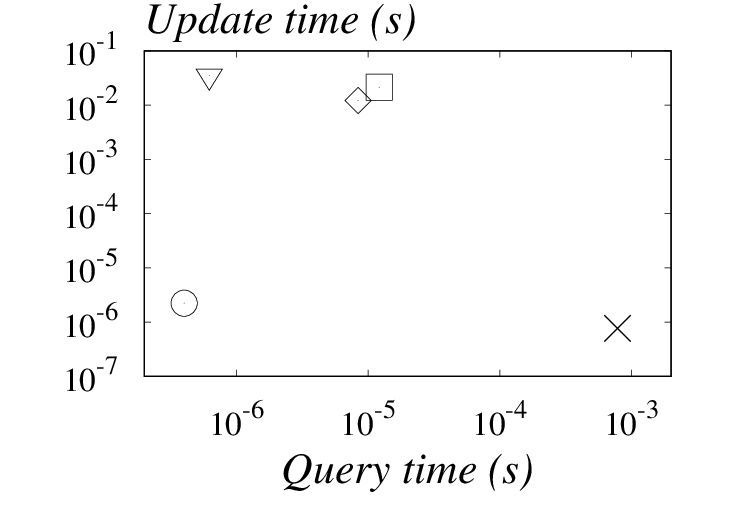}&
				\hspace{-2mm} \includegraphics[height=26mm]{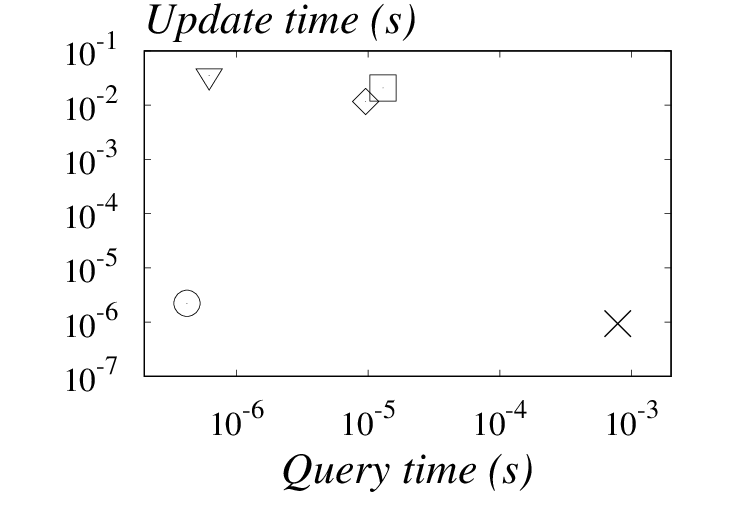}
				\\[-2mm]
				\hspace{-1mm} (a) Exponential distribution &
				\hspace{-1mm} (b) Normal distribution &
				\hspace{-1mm} (c) Half-normal distribution&
				\hspace{-1mm} (d) Log-normal distribution \\[-1mm]
			\end{tabular}
			\vspace{-2mm}
			\caption{Update time v.s. query time on different distributions (in seconds). ($n=10^5$)} \label{fig:qvsu}
			\vspace{-2mm}
		\end{small}
	\end{figure*}

	\subsection{Query and Update Efficiency Trade-off}
	Given the pressing needs to balance sample and update efficiency in the \pps~problem, we present trade-off plots between query time and update time for each method in Fig.~\ref{fig:qvsu}. To illustrate this trade-off more clearly, we include a brute-force method in our comparison. This brute-force method stores all elements and their weights in dynamic arrays and performs queries through a brute-force scan. Basically, this brute-force method has the lowest possible update cost, which only needs to dynamically maintain the set.

	We set the number of elements to $n=10^5$. As shown in  Fig.~\ref{fig:qvsu}, there is a noticeable trade-off between update efficiency and query efficiency among the competitor methods. Our method, \dub, achieves the best query efficiency among all methods while maintaining superb update efficiency, comparable even to the brute-force approach that dynamically maintains the set. Remarkably, \dips~ achieves a four orders of magnitude speed-up in index update efficiency compared to \rhss, \rbss, and \rdss. This confirms that \dips~ offers the best balance between query efficiency and index update efficiency, in line with our theoretical analysis showing that \dips~ achieves optimal expected query and update efficiency.


	\subsection{Query Efficiency}
	To further assess the query efficiency of these algorithms, we vary the number of elements, $n$, and report the query time in Fig.~\ref{fig:nvsq}. The values of $n$ tested are $\{10^6, 5\times10^6,10^7, 5\times10^7,10^8, 5\times10^8,10^9\}$. For each method, we report the average query time based on 10,000 queries. Notice that we omit the brute-force approach, which requires $\bigo{n}$ query time, in this set of experiment since it is too slow to finish these 10,000 queries. As shown in Fig.~\ref{fig:nvsq}, our method and \rdss~ demonstrate comparable query efficiency, matching their optimal expected query efficiency, and both outperform other competitors significantly. We also vary $c$ to 0.8, 0.6, and 0.4, and show the results in Fig.~\ref{fig:nvsq8}, Fig.~\ref{fig:nvsq6}, and Fig.~\ref{fig:nvsq4}.
	\begin{figure*}[t]
		\centering
		\begin{small}
			\begin{tabular}{cccc}
					\multicolumn{4}{c}{\hspace{-6mm} \includegraphics[height=2.3mm]{figure/nvsu_normal_legend.eps}}  \\[-2mm]
				\hspace{-2mm} \includegraphics[height=26mm]{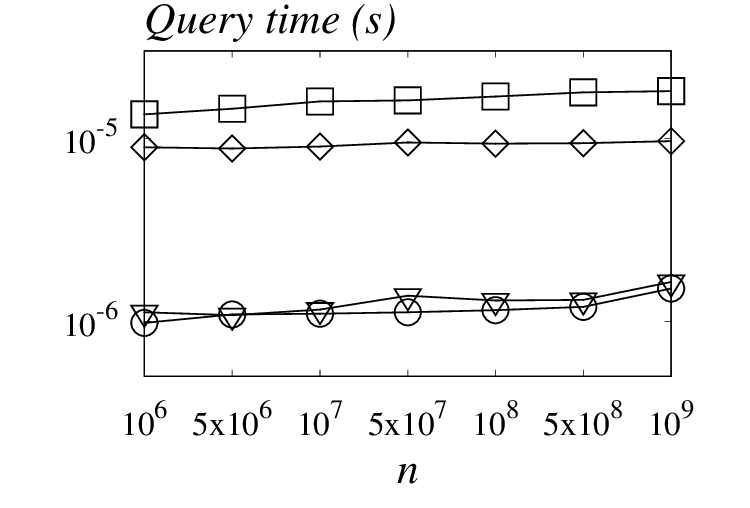} &
				\hspace{-2mm} \includegraphics[height=26mm]{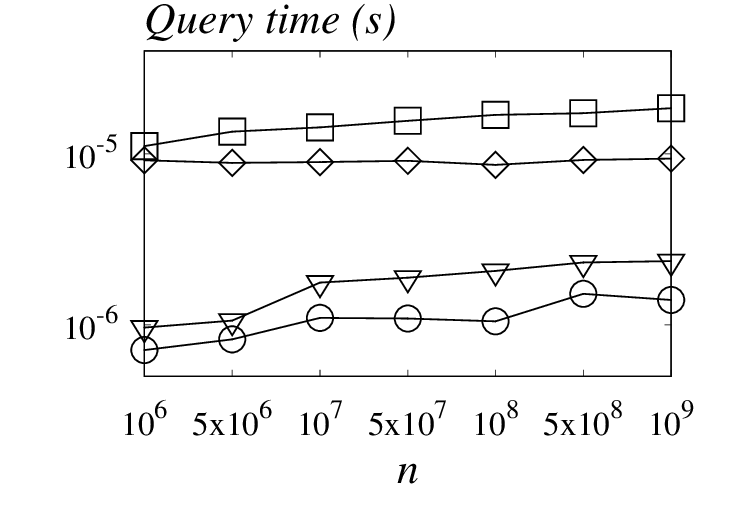} &
				\hspace{-2mm} \includegraphics[height=26mm]{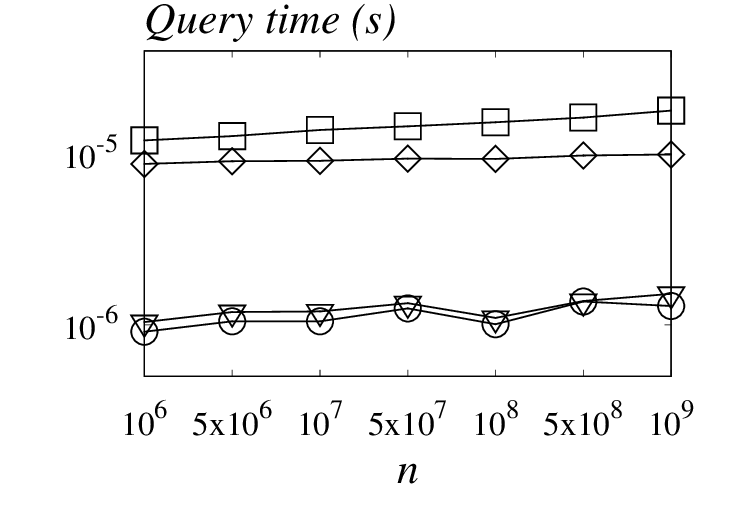}&
				\hspace{-2mm} \includegraphics[height=26mm]{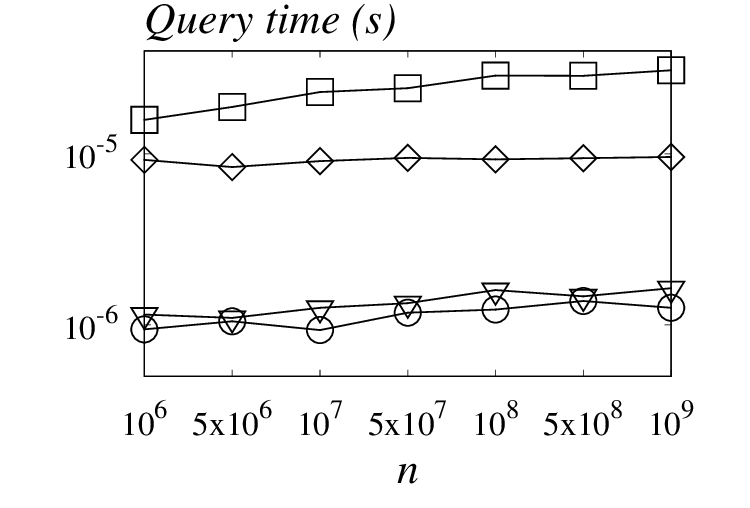}
				\\[-2mm]
				\hspace{-1mm} (a) Exponential distribution &
				\hspace{-1mm} (b) Normal distribution &
				\hspace{-1mm} (c) Half-normal distribution&
				\hspace{-1mm} (d) Log-normal distribution \\[-1mm]
			\end{tabular}
			\vspace{-2mm}
			\caption{Varying $n$: Query time (in seconds) on different distributions. (c=1)} \label{fig:nvsq}
			\vspace{-2mm}
		\end{small}
	\end{figure*}

	\begin{figure*}[t]
		\centering
		\begin{small}
			\begin{tabular}{cccc}
					\multicolumn{4}{c}{\hspace{-6mm} \includegraphics[height=2.3mm]{figure/nvsu_normal_legend.eps}}  \\[-2mm]
				\hspace{-2mm} \includegraphics[height=26mm]{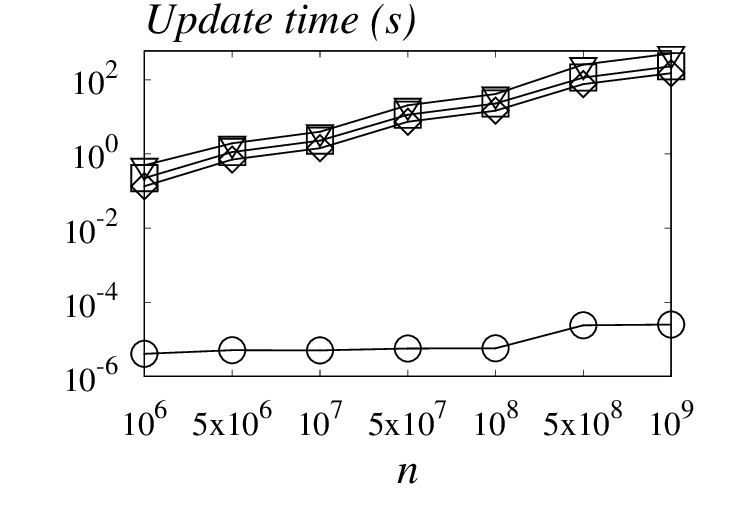} &
				\hspace{-2mm} \includegraphics[height=26mm]{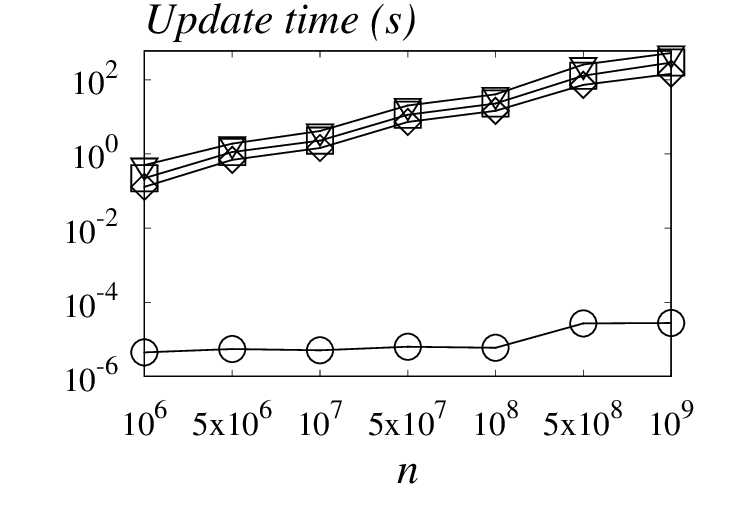} &
				\hspace{-2mm} \includegraphics[height=26mm]{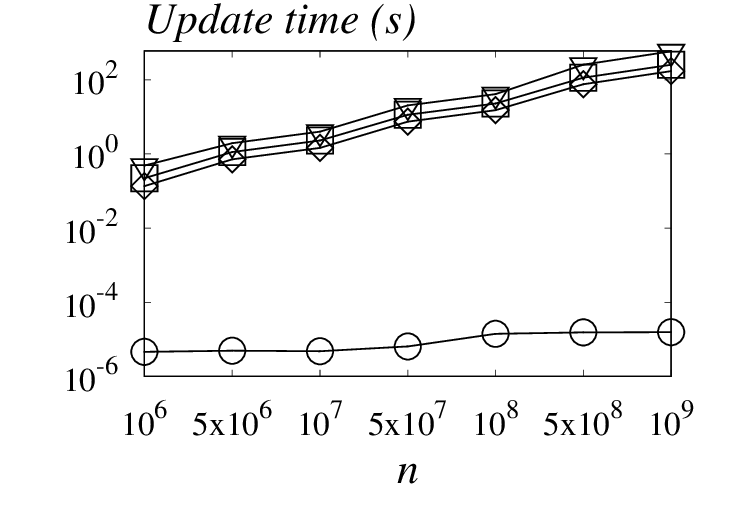}&
				\hspace{-2mm} \includegraphics[height=26mm]{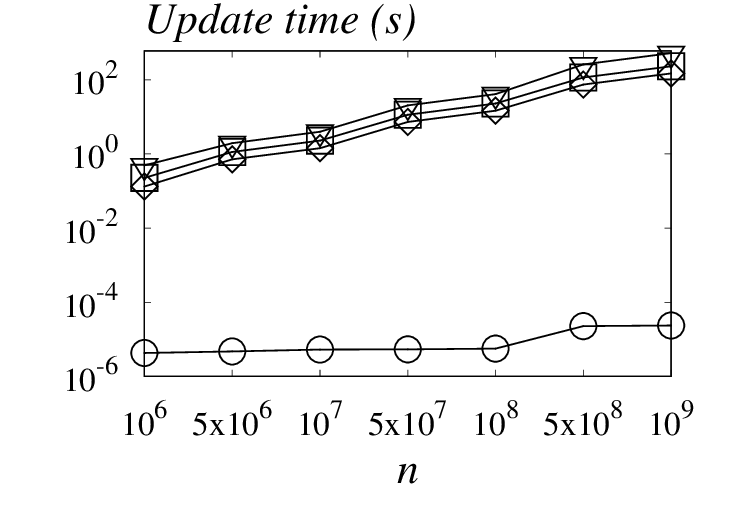}
				\\[-2mm]
				\hspace{-1mm} (a) Exponential distribution &
				\hspace{-1mm} (b) Normal distribution &
				\hspace{-1mm} (c) Half-normal distribution&
				\hspace{-1mm} (d) Log-normal distribution \\[-1mm]
			\end{tabular}
			\vspace{-2mm}
			\caption{Varying $n$: Update time (in seconds) on different distributions.} \label{fig:nvsu}
			\vspace{-2mm}
		\end{small}
	\end{figure*}

	\subsection{Update Efficiency}
	To evaluate the update efficiency of these algorithms, we vary the number $n$ of elements, using values from $\{10^6, 5\times10^6,10^7, 5\times10^7,10^8, 5\times10^8,10^9\}$, and presented the update times in Fig.~\ref{fig:nvsu}. For our method \dips, the average update time is reported based on 100,000 insertions and 100,000 deletions. During our experiments, we observe that the update time for competitor methods were prohibitively high, since the update cost is linear to $n$. However, the update time remained stable across different trials. Consequently, for competitor methods, we report the average update time based on 50 insertions and 50 deletions. As shown in Fig.~\ref{fig:nvsu}, our method, \dub, consistently outperforms its competitors by at least four orders of magnitude across all weight distributions. Moreover, as $n$ increases, the performance advantage of our method becomes more pronounced. This is because the competitors reply on reconstruction to handle insertions or deletions, while the update time of \dub~ remains $\bigone$ regardless of $n$.

	\begin{figure*}[t]
		\centering
		\begin{small}
			\begin{tabular}{cccc}
					\multicolumn{4}{c}{\hspace{-6mm} \includegraphics[height=3.1mm]{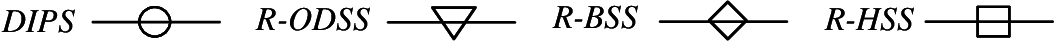}}  \\[0mm]
				\includegraphics[height=26mm]{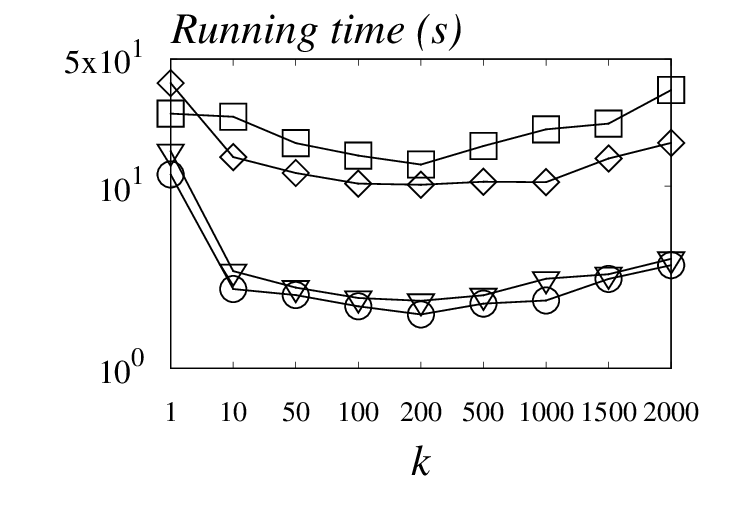} &
				\hspace{-2mm} \includegraphics[height=26mm]{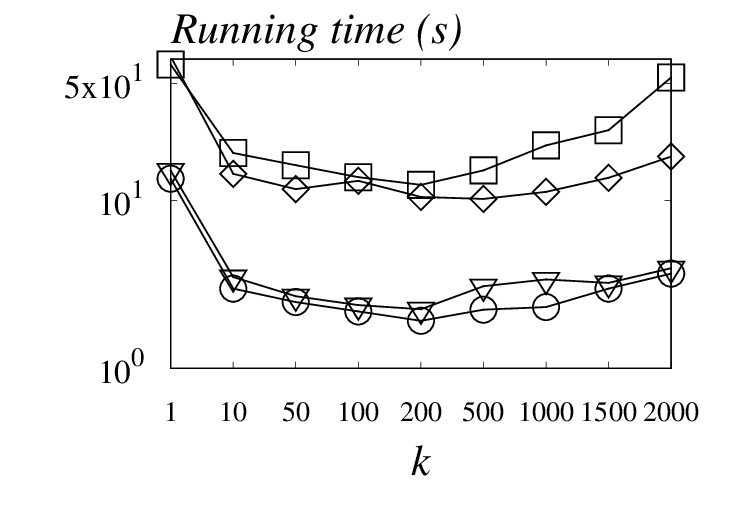}&
				\hspace{-2mm} \includegraphics[height=26mm]{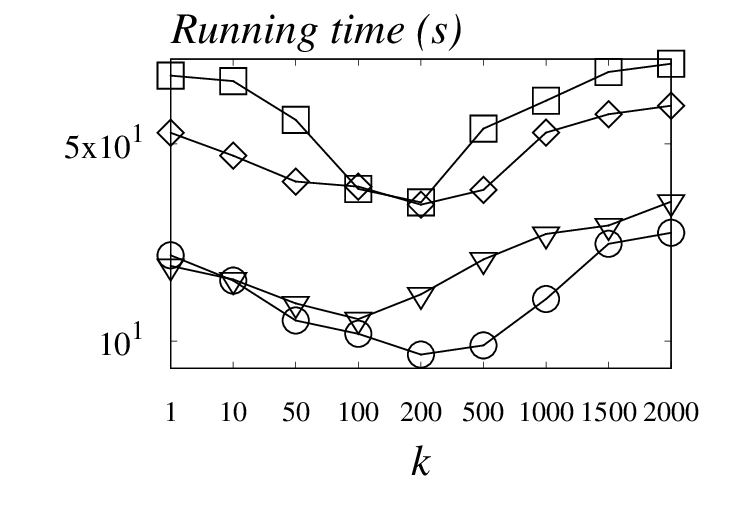} &
				\hspace{-2mm} \includegraphics[height=26mm]{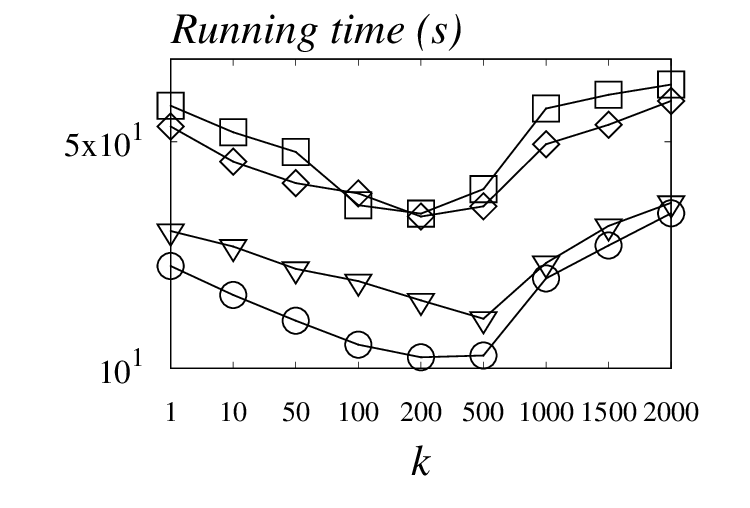}
				\\[-2mm]
				\hspace{-1mm} (a) OL - Exponential distribution &
				\hspace{-1mm} (b) OL - Weibull distribution &
				\hspace{-1mm} (c) TW - Exponential distribution &
				\hspace{-1mm} (d) TW - Weibull distribution \\[-1mm]
			\end{tabular}
			\vspace{-2mm}
			\caption{Running time of dynamic IM algorithm based on different \pps~indexes.} \label{fig:imtime}
			\vspace{-2mm}
		\end{small}
	\end{figure*}

	\begin{figure*}[t]
		\centering
		\begin{small}
			\begin{tabular}{cccc}
					\multicolumn{4}{c}{\hspace{-6mm} \includegraphics[height=2.3mm]{figure/imtime_ol_exp_legend.eps}}  \\[-2mm]
				\hspace{-2mm} \includegraphics[height=26mm]{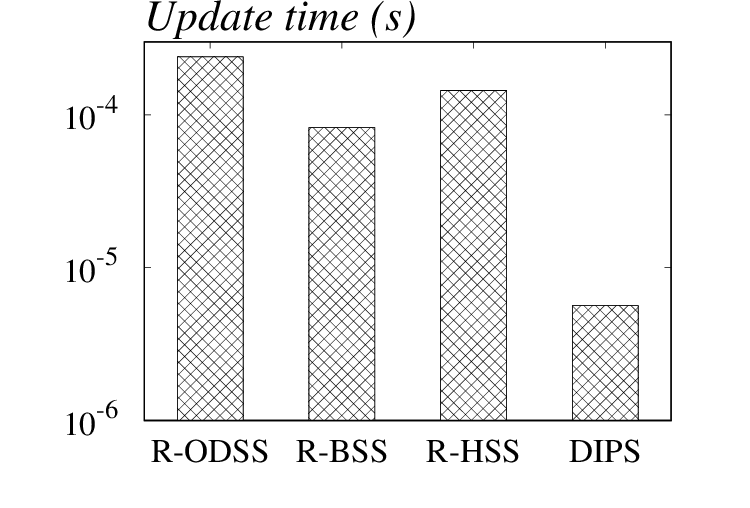} &
				\hspace{-2mm} \includegraphics[height=26mm]{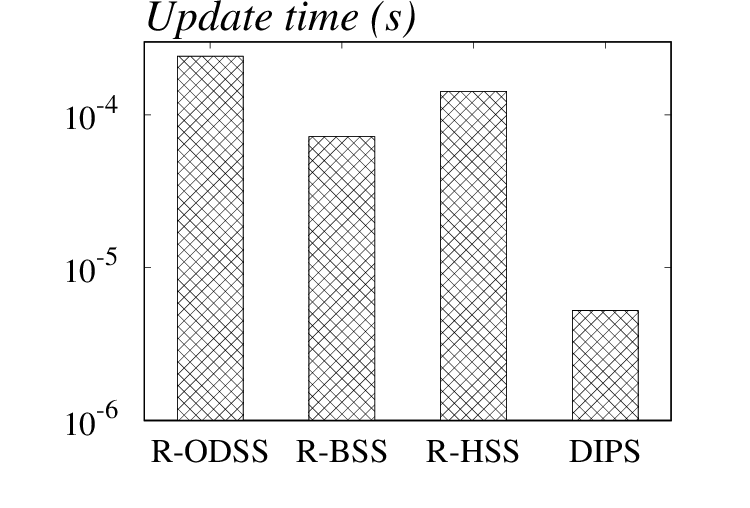}&
				\hspace{-2mm} \includegraphics[height=26mm]{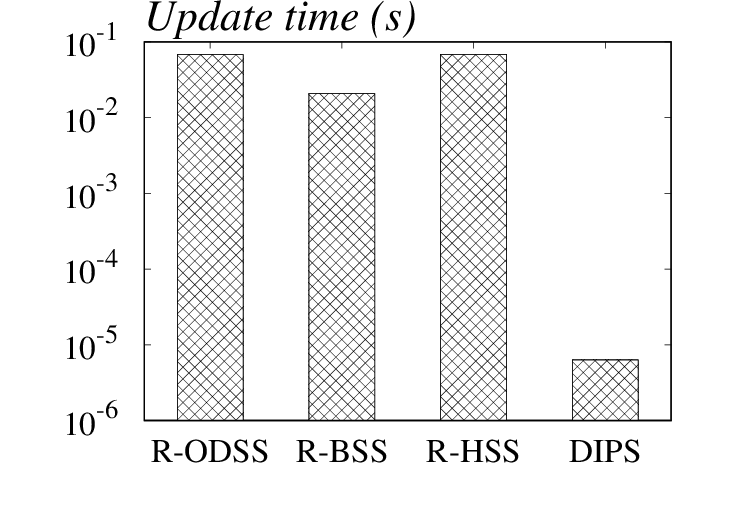} &
				\hspace{-2mm} \includegraphics[height=26mm]{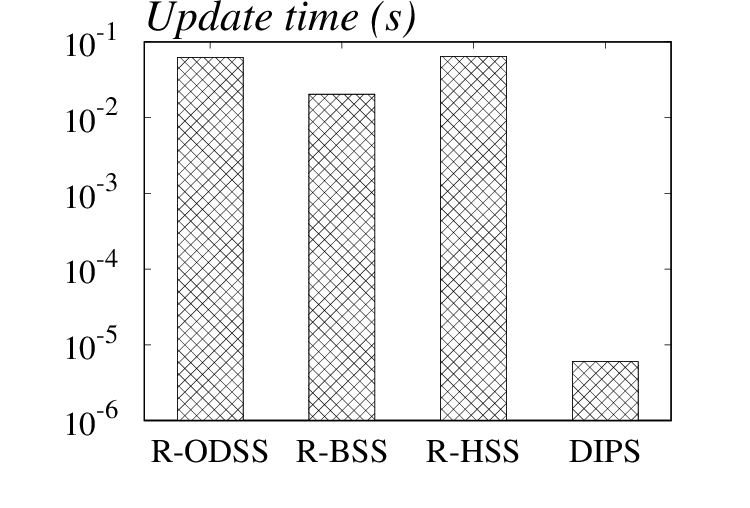}
				\\[-2mm]
				\hspace{-1mm} (a) OL - Exponential distribution &
				\hspace{-1mm} (b) OL - Weibull distribution &
				\hspace{-1mm} (c) TW - Exponential distribution &
				\hspace{-1mm} (d) TW - Weibull distribution \\[-1mm]
			\end{tabular}
			\vspace{-2mm}
			\caption{Update time of dynamic IM algorithm based on different \pps~indexes.} \label{fig:imupd}
			\vspace{-2mm}
		\end{small}
	\end{figure*}

	\subsection{Memory Usage}
        Table~\ref{tab:mem} shows the memory usage of DIPS and R-ODSS for different input sizes ranging from $10^6$ to $10^9$ elements. Both methods exhibit linear space complexity as expected from their theoretical analysis. While DIPS uses slightly more memory than R-ODSS (about 1.1-1.6x), the difference is relatively small and remains consistent across different input sizes. For example, with $n=10^9$ elements, DIPS uses 65.38GB compared to R-ODSS's 64.23GB, representing only a 1.8\% increase. This small overhead is acceptable given DIPS's improved update performance. The other competitors (R-HSS and R-BSS) are omitted from the comparison as they use similar or more space than R-ODSS while providing worse time complexity guarantees.
        \begin{table}[t]\centering
            \caption{Memory usage.}
            \label{tab:mem}
            \scalebox{0.8}{\begin{tabular}
            {|c|c|c|c|c|c|}
                \hline \diagbox{Method}{Memory Usage (GB)}{n} 
                    & $10^6$            & $10^7$            & $10^8$             & $10^9$                
                \\\hline
                DIPS &$0.24$            & $0.84$            & $6.08$             & $65.38$
            \\\hline
                R-ODSS &$0.15$            & $0.71$            & $5.95$             & $64.23$
            \\\hline
            \end{tabular}}
        \end{table}
        

\section{Application to IM in Evolving Graphs}\label{sec:dynamic-im}
    In this section, we apply our method, \dips, to the problem of Influence Maximization (IM) in evolving graphs. Influence Maximization is a fundamental data mining task that has been extensively studied for decades. Given a positive integer $k$ and a cascade model, the goal is to identify a set of $k$ nodes, also called seed nodes, in a network that can maximize the expected spread of influence, such as information or behaviors, throughout the network. Weighted Cascade (WC) model is one of the most popular cascade model widely used and tested in the literature. In this model, each edge $\langle u,v \rangle$ is assigned with a weight $w(u,v)$, and the probability $p(u,v)$ associated with each edge $\langle u,v \rangle$ is $ \frac{w(u,v)}{\sum_{x\in IN(v)}w(x,v)}$, where $IN(v)$ is the set of in-neighbors of vertex $v$. A simple setting in WC is to make all weights equal. In this case, $p(u,v)=\frac{1}{d_{in}(v)}$. However, existing study shows that such a setting ignores the fact that given a vertex $v$, the importance of its incoming neighbors may actually differ drastically. Thus, existing studies, e.g., \cite{GoyalBL10,KutzkovBBG13}, suggest setting the weights and hence the probabilities differently. 

    For IM, the most popular approaches use the idea of random Reverse Reachable (RR) sets \cite{BorgsBCL14,Guo0WC20Subsim,GuoFZW23,FengCGZW24,GuoWWLT22,BianGWY20}. For WC model, a random RR set is generated as follows: {\em (i)} A vertex $t$ is randomly chosen as the target vertex and marked as ``visited""; {\em (ii)} For each visited vertex $v$, it has a single chance to activate each of its incoming neighbor $u$ following the probability $p(u,v)$, and if $u$ is marked as activated it is marked also as visited; {\em (iii)} The process repeats until all visited nodes have finished their single chance to active their incoming neighbors. The step (ii) is essentially a \pps~ problem. For the RR set-based solutions, they sample a large number of random RR sets and then find the seed set based on the sampled random RR sets.

    In the context of evolving graphs, where network structures change over time, the dynamic IM problem extends the traditional IM framework by accounting for these temporal changes. This requires strategies that can adapt to the evolving network topology, which is crucial in dynamic environments such as social media platforms. In fully dynamic network models, where users can join or leave the network and influence propagation can change in real-time, traditional static IM algorithms fall short. According to research by Peng \cite{Peng}, no algorithm can guarantee a meaningful approximation in such dynamic settings without incurring significant computational costs. Re-running an IM algorithm upon each update achieves the running time lower bound, indicating that no more efficient methods are available under current theoretical limits.

    Existing static IM algorithms are designed for static networks and require complete rebuilding of all relevant structures when the network changes. This process is costly and inefficient. If the structures needed for running the IM algorithm could be updated incrementally rather than rebuilt from scratch, the query time could be greatly improved. This insight motivates the need for new approaches that can efficiently handle changes in network topology.

    In the following, we conduct experiments to evaluate the performance by integrating the above sampling structures into the state-of-the art IM algorithm SUBSIM \cite{Guo0WC20Subsim}, focusing on both the running time for IM and the update time required for edge insertions and deletions. The experiments were conducted on two real-world graphs, Orkut-Links (OL) and Twitter (TW), both of which are publicly available datasets~\cite{snapnets}. The Orkut graph consists of 3,072,441 nodes and 117,185,083 edges, and it is undirected. The Twitter graph is directed and comprises 41,652,230 nodes and 1,146,365,182 edges. Following \cite{Guo0WC20Subsim,ODSS}, for each dataset, we examined scenarios where the edge weights follow two types of distributions: the exponential distribution and the Weibull distribution. For the exponential distribution, we set the parameter $\lambda = 1$. For the Weibull distribution, the probability density function is given by:
    $f(x) = \frac{b}{a} \left(\frac{x}{a}\right)^{b-1} e^{-(x/a)^b}$
    where $a$ is the scale parameter and $b$ is the shape parameter. For each edge in the graph, the parameters $a$ and $b$ are drawn uniformly from interval [0, 10].

    \subsection{Running Time of Dynamic IM Algorithms}
        In our experiments, we vary the size $k$ of the seed set across the values $\{1, 10, 50, 100, 200, 500, 1000, 2000\}$. Each algorithm is repeated five times to generate the seed set, and we reported the average running time. The results are presented in Fig.~\ref{fig:imtime}. As shown in the figure, our method demonstrates competitive running efficiency as that using \rdss~ and consistently outperforms the remaining competitors as the sampling efficiency of our \dips~ and \rdss~ are optimal. This demonstrates the superb efficiency by integrating \dips~ into IM algorithms.  
    \subsection{Update Time of Dynamic IM Algorithms}
        To simulate the evolution of social networks, we uniformly select $10^6$ edges from each graph and measure the average time required to update the sampling structures for the insertion and deletion of these edges. Fig.~\ref{fig:imupd} presents the update times of all algorithms. On the OL dataset, \dub~ outperforms its competitors by at least one order of magnitude. On the larger TW dataset, \dub~ outperforms its competitors by at least three orders of magnitude. This shows the superb efficiency of our \dips~ framework. 


\section{Conclusion}\label{sec:conclusion} 
    In this paper, we study the \pps~problem under the dynamic setting.
    We propose DIPS, an optimal dynamic index for \pps. Our DIPS achieves an expected query time of $\bigone$ and an expected update time of $\bigone$, with space consumption of $\bigon$. Experimental results suggest that DIPS significantly outperforms its competitors in terms of update efficiency while maintaining competitive query time compared to states of the art.
    

\balance

\bibliographystyle{ACM-Reference-Format}
\bibliography{reference}


\appendix
\section{Pseudocode and Proof of Lem.~ \ref{wss_grs}}\label{apd:grs}
     The pseudocode of the query and update algorithms for the bounded weight ratio (corresponding to Lem.~ \ref{wss_grs}) is provided in Alg. \ref{alg:bwr}. 

 \header{\bf Proof of Lem.~ \ref{wss_grs}}.
        We will maintain a dynamic array $A$ in arbitrary order and a hash table to map from each element $v\in T$ to its location in array $A$. Note that in this paper, we can use the standard background rebuilding technique~\cite{rebuilding} to eliminate the amortized cost associated with the rebuilding process. We further maintain the size $t$ of $T$, the total weight $W_T$ of all elements in $T$. To answer queries, let us first select candidates by overestimating all probabilities to be $\bar{p}=\frac{c\cdot \bar{w}}{W_\S}$. To do so, we first generate a random number $u$ from $\unif{0}{1}$. If $u\leq W_T/W_\S\triangleq q$, we locate the first candidate by setting $j=\lf\frac{\log(1 - q\cdot\unif{0}{1})}{\log(1 - \bar{p})}\rf$, after this, we can always find the next candidate by jumping to the next $g = \lf\frac{\log(1 - \unif{0}{1})}{\log(1 - \bar{p})}\rf$-th element. We keep jumping until we reach the end of the array $A$. Because we may grant each candidate $v$ with extra probability, we should draw a number $u$ from $\unif{0}{\bar{p}}$ and finally include $v$ in the sample set $X$ if and only if $u\leq\frac{c\cdot w}{W_\S}$. In this way, each element $v\in\S$ is sampled into $X$ with probability $\bar{p}\times\frac{c\cdot w}{\bar{p}\cdot W_\S}=\frac{c\cdot w}{W_\S}$.

 \begin{algorithm}[h!]
        \caption{Bounded weight ratio}
        \label{alg:bwr}
        \DontPrintSemicolon
        
        \SetKwProg{Init}{init$(\S, w, T)$}{:}{}
        \Init{} {
            $A \gets$ new dynamic array;
            $H \gets$ new hash table\;
            $t \gets |T|$;
            $W_T \gets 0$\;
            \ForEach{$v \in T$}{
                Append $v$ to $A$\;
                $H[v] \gets$ position of $v$ in $A$\;
                $W_T \gets W_T + w(v)$\;
            }
        }
        
        \SetKwProg{Cw}{change\_w$(v, w)$}{:}{}
        \Cw{} {
            $W_T \gets W_T - w(v) + w$; $w(v) \gets w$\;
        }
        
        \SetKwProg{Ins}{insert$(v, w)$}{:}{}
        \Ins{} {
            Append $v$ to $A$\;
            $H[v] \gets$ position of $v$ in $A$\;
            $W_T \gets W_T + w$; $t \gets t + 1$\;
        }
    
        \SetKwProg{Del}{delete$(v)$}{:}{}
        \Del{} {
            $i \gets H[v]$\;
            $W_T \gets W_T - w(v)$\;
            Remove $v$ from $H$, swap $A[i]$ with last element in $A$, update position in $H$ for swapped element, and remove the last element from $A$\;
            $t \gets t - 1$\;
        }
    
        \SetKwProg{Que}{query$()$}{:}{}
        \Que{} {
            $X \gets \emptyset$; $u \gets \unif{0}{1}$; $q \gets W_T/W_\S$\;
            \If{$u \leq q$}{
                $\bar{p} \gets c\cdot\bar{w}/W_\S$\;
                $j \gets \lf\frac{\log(1 - q\cdot\unif{0}{1})}{\log(1 - \bar{p})}\rf$\;
                \While{$j < t$}{
                    $v \gets A[j]$\;
                    $u \gets \unif{0}{\bar{p}}$\;
                    \If{$u \leq c\cdot w(v)/W_\S$}{
                        $X \gets X \cup \{v\}$\;
                    }
                    $g \gets \lf\frac{\log(1 - \unif{0}{1})}{\log(1 - \bar{p})}\rf$\;
                    $j \gets j + g$\;
                }
            }
            \Return{$X$}\;
        }
    \end{algorithm}

            We next analyze the query time.
        Suppose we have located the first candidate at position $j$, let $g_1, g_2, \ldots$ be the sequence of random numbers we generate in the remaining process. Let $t$ be the smallest number such that $\sum_{i=1}^{t+1} g_i\geq n-j$. On the one hand, $t$ can be interpreted as the number of jumps before reaching the end. On the other hand, $t$ can be interpreted as the number of success trails among $n-j$ independent trails each with success probability $\bar{p}$. We can take the later interpretation and derive that $t\sim\bin{n-j}{\bar{p}}$. The expected number of jumps is thus
        $$E[t]=(n-j)\bar{p}\leq n\bar{p}=\sum_{v\in\S}\frac{c\cdot\bar{w}}{W_\S}\leq b\sum_{v\in\S}\frac{c\cdot w}{W_\S}\leq bc.$$ 
        So the expected query time is $q(1+E[t])+(1-q)=\bigo{1}$.

       \begin{algorithm}[t!]
        \caption{Update operation of Lemma~\ref{wss_sr}}
        \label{alg:sr_update}
        \DontPrintSemicolon

        \SetKwProg{Cw}{change\_w$(v, w_{new})$}{:}{}
        \Cw{} {
            $j_{old} \gets \lfloor \log_b w(v) \rfloor$; $j_{new} \gets \lfloor \log_b w_{new} \rfloor$\;
            \If{$j_{old} \neq j_{new}$}{
                Remove $v$ from $B_{j_{old}}$ and add $v$ to $B_{j_{new}}$\;
                Update $D(B_{j_{old}})$ and $D(B_{j_{new}})$\;
            }
            Update weights $w(B_{j_{old}}), w(B_{j_{new}}), w(C_{t_{old}}), w(C_{t_{new}})$\;
            Update $w'(B_{j_{old}}), w'(B_{j_{new}}), w'(C_{t_{old}}), w'(C_{t_{new}})$\;
            Update $\widetilde{D}(C_{t_{old}})$ and $\widetilde{D}(C_{t_{new}})$\;
            $W_\S \gets W_\S - w(v) + w_{new}$; $w(v) \gets w_{new}$\;
        }
        
        \SetKwProg{Ins}{insert$(v, w)$}{:}{}
        \Ins{} {
            \If{$Size \geq 2 \cdot OldSize$}{
                Rebuild entire structure\;
                \Return\;
            }
            $j \gets \lfloor \log_b w(v) \rfloor$\;
            Add $v$ to $B_j$ and  update $D(B_j)$\;
            Update weights $w(B_j), w(C_t), w'(B_j), w'(C_t)$\;
            Update $\widetilde{D}(C_t)$\;
            $W_\S \gets W_\S + w$; $Size \gets Size + 1$\;
        }

        \SetKwProg{Del}{delete$(v)$}{:}{}
        \Del{} {
            \If{$Size \leq OldSize/2$}{
                Rebuild entire structure\;
                \Return\;
            }
            $j \gets \lfloor \log_b w(v) \rfloor$\;
            Remove $v$ from $B_j$ and update $D(B_j)$\;
            Update weights $w(B_j), w(C_t), w'(B_j), w'(C_t)$\;
            Update $\widetilde{D}(C_t)$\;
            $W_\S \gets W_\S - w(v)$; $Size \gets Size - 1$\;
        }
        \end{algorithm}




\section{Example of Buckets and Chunks}\label{apd:example}
\begin{example}
    Consider a set of elements with $n=100$ and $b=2$. Suppose we have elements with weights 3.5, 5.2, 6.8 (falling in $(2^1,2^2]$), 17.3, 19.1 (falling in $(2^4,2^5]$), and 33.2 (falling in $(2^5,2^6]$). These elements would be placed into buckets $B_1$, $B_4$, and $B_5$ respectively. The total weights of these buckets would be $w(B_1)=15.5$, $w(B_4)=36.4$, and $w(B_5)=33.2$. With $\lc\log_2 100\rc=7$, bucket $B_1$ would belong to chunk $C_0$ (since $1\in[0,6]$), while buckets $B_4$ and $B_5$ would belong to chunk $C_1$ (since $4,5\in[7,13]$).
\end{example}

\section{Update Operations of Lemma~\ref{wss_sr}}\label{apd:sr_update}
                The pseudocode is provided in Alg.~\ref{alg:sr_update}. We need to reconstruct the entire data structure every time the size of the set $\S$ doubles or halves since the last construction. The amortized cost for reconstruction is $\bigone$. Recap that we can implement the background rebuilding technique for the rebuilding process, thus this amortized cost can be eliminated. If an update operation does not incur reconstruction, its cost is also $\bigone$ because it changes only one chunk and one bucket for \ins~ and \del~ and at most two chunks and two buckets for \upd.  Thus, it only takes constant time to update these information. Also, it only takes constant time to update $W$ and $Size$.
 

    \begin{figure*}[t]
    	\centering
    	\begin{small}
    		\begin{tabular}{cccc}
    	 			\multicolumn{4}{c}{\hspace{-6mm} \includegraphics[height=2.3mm]{figure/nvsu_normal_legend.eps}}  \\[-2mm]
    			\hspace{-2mm} \includegraphics[height=26mm]{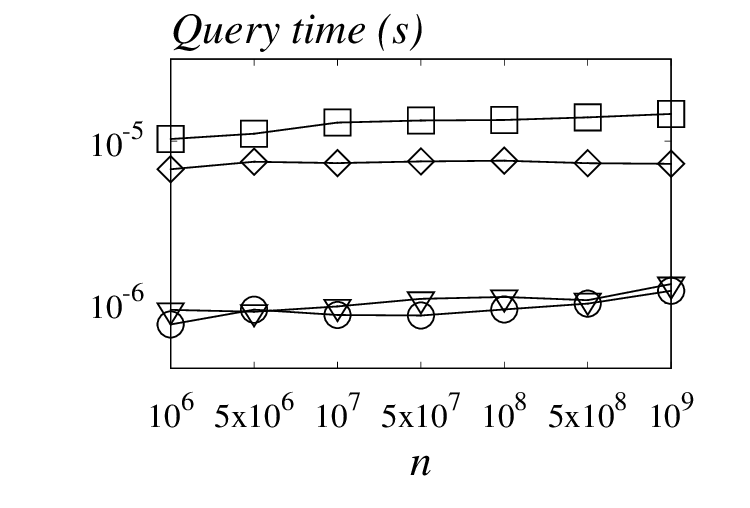} &
    			\hspace{-2mm} \includegraphics[height=26mm]{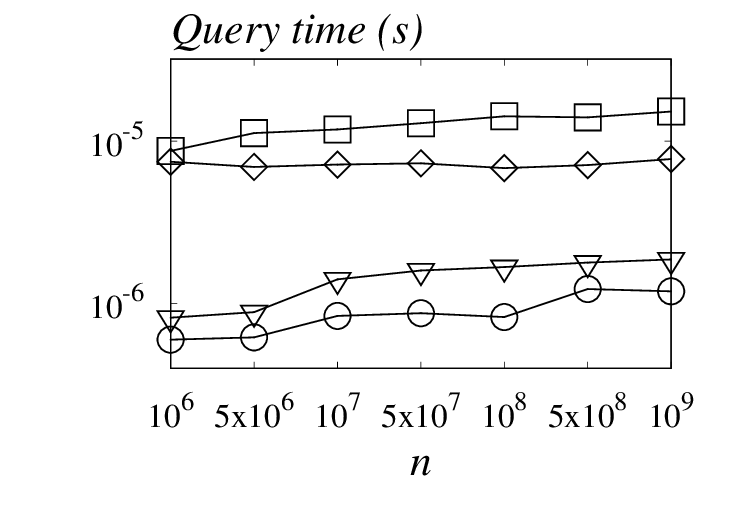} &
    			\hspace{-2mm} \includegraphics[height=26mm]{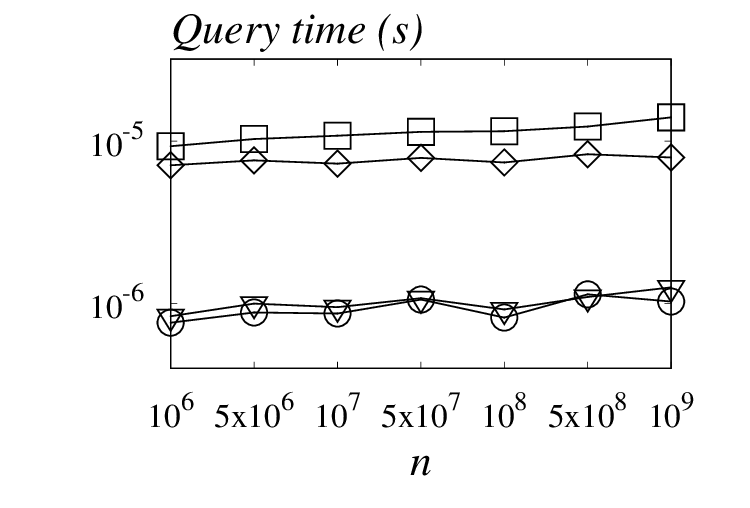}&
    			\hspace{-2mm} \includegraphics[height=26mm]{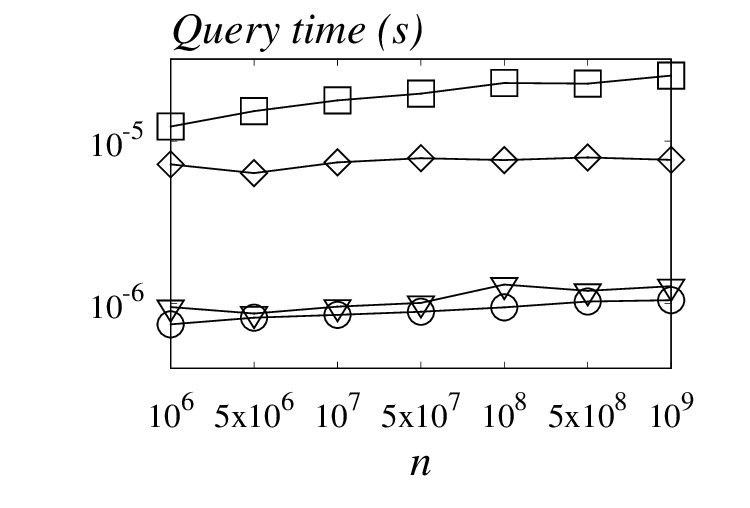}
    			\\[-2mm]
    			\hspace{-1mm} (a) Exponential distribution &
    			\hspace{-1mm} (b) Normal distribution &
    			\hspace{-1mm} (c) Half-normal distribution&
    			\hspace{-1mm} (d) Log-normal distribution \\[-1mm]
    		\end{tabular}
    		\vspace{-2mm}
    		\caption{Varying $n$: Query time (in seconds) on different distributions. (c=0.8)} \label{fig:nvsq8}
    		\vspace{-2mm}
    	\end{small}
    \end{figure*}
    
    \begin{figure*}[t]
    	\centering
    	\begin{small}
    		\begin{tabular}{cccc}
    	 			\multicolumn{4}{c}{\hspace{-6mm} \includegraphics[height=2.3mm]{figure/nvsu_normal_legend.eps}}  \\[-2mm]
    			\hspace{-2mm} \includegraphics[height=26mm]{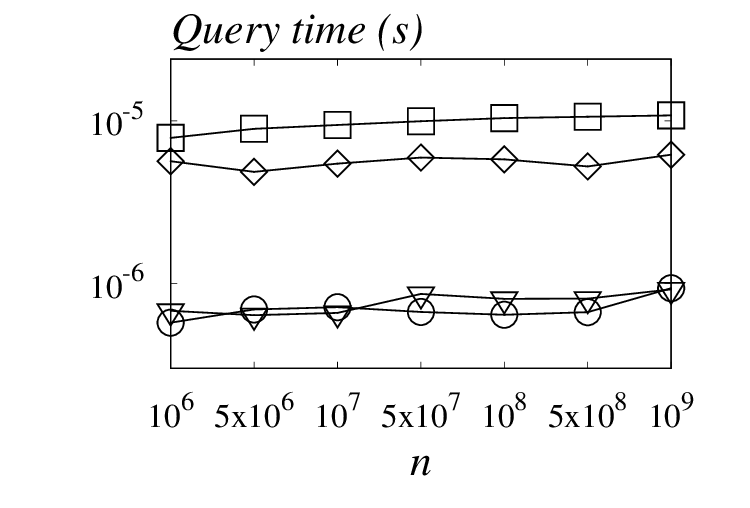} &
    			\hspace{-2mm} \includegraphics[height=26mm]{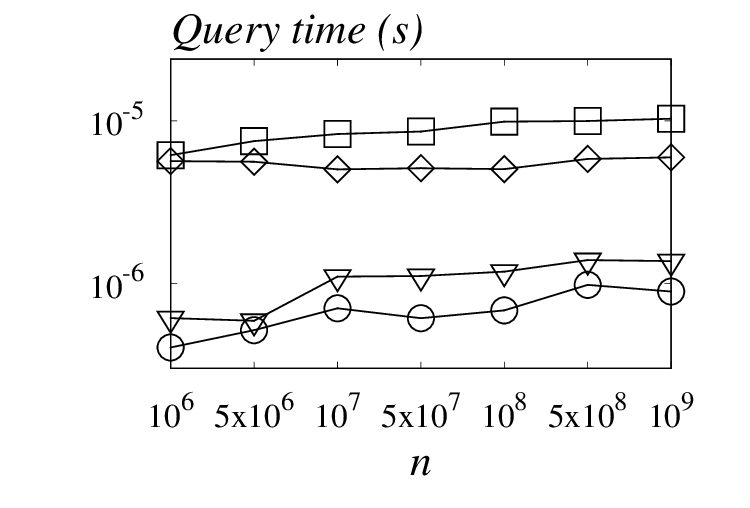} &
    			\hspace{-2mm} \includegraphics[height=26mm]{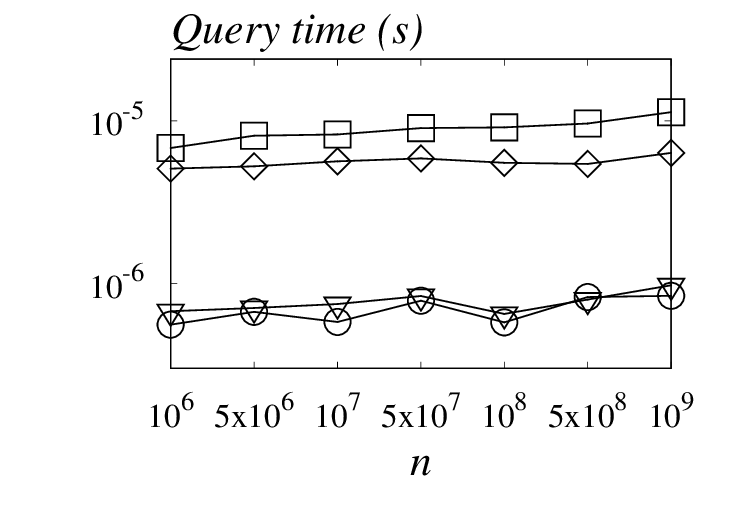}&
    			\hspace{-2mm} \includegraphics[height=26mm]{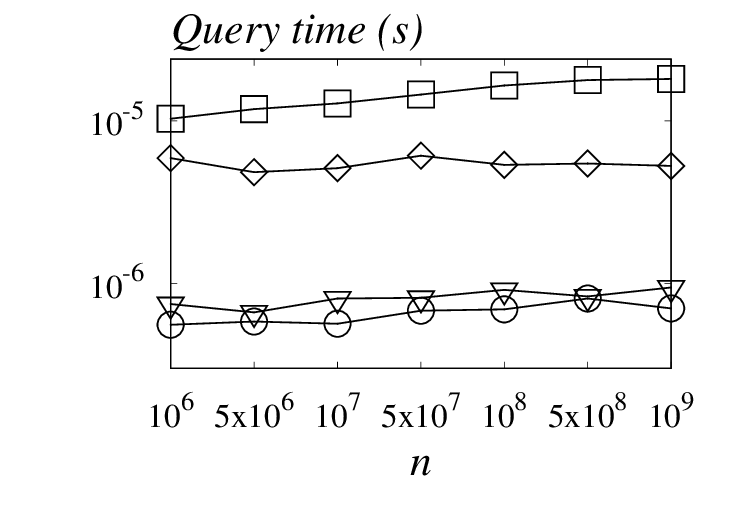}
    			\\[-2mm]
    			\hspace{-1mm} (a) Exponential distribution &
    			\hspace{-1mm} (b) Normal distribution &
    			\hspace{-1mm} (c) Half-normal distribution&
    			\hspace{-1mm} (d) Log-normal distribution \\[-1mm]
    		\end{tabular}
    		\vspace{-2mm}
    		\caption{Varying $n$: Query time (in seconds) on different distributions. (c=0.6)} \label{fig:nvsq6}
    		\vspace{-2mm}
    	\end{small}
    \end{figure*}
    
    \begin{figure*}[t]
    	\centering
    	\begin{small}
    		\begin{tabular}{cccc}
    	 			\multicolumn{4}{c}{\hspace{-6mm} \includegraphics[height=2.3mm]{figure/nvsu_normal_legend.eps}}  \\[-2mm]
    			\hspace{-2mm} \includegraphics[height=26mm]{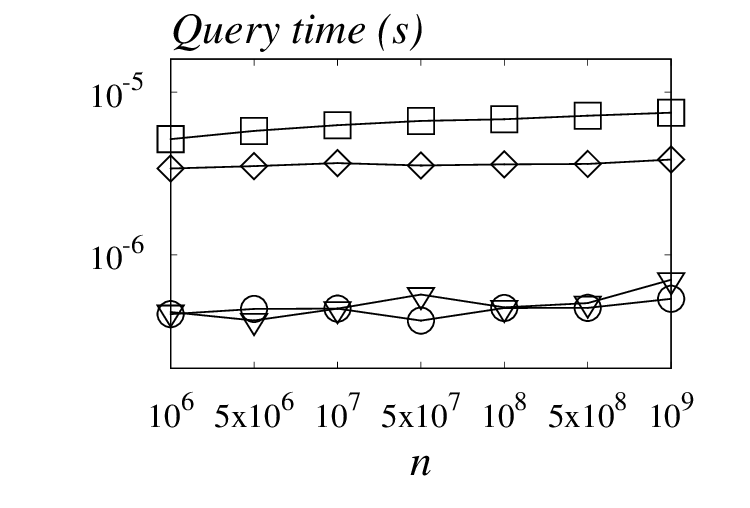} &
    			\hspace{-2mm} \includegraphics[height=26mm]{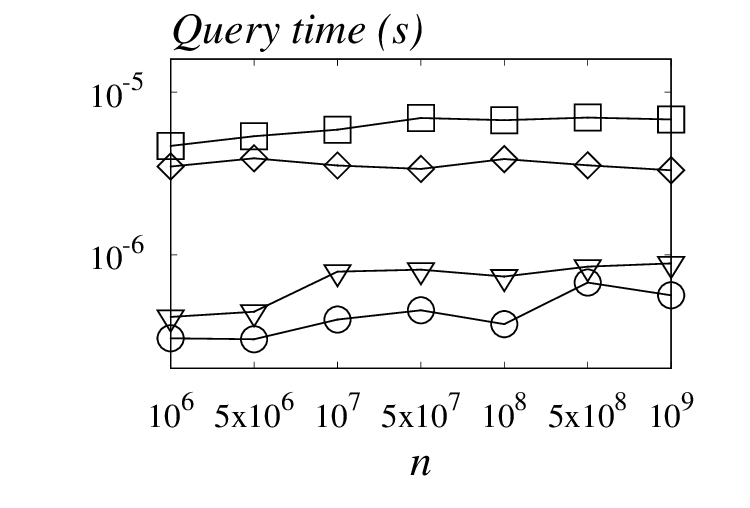} &
    			\hspace{-2mm} \includegraphics[height=26mm]{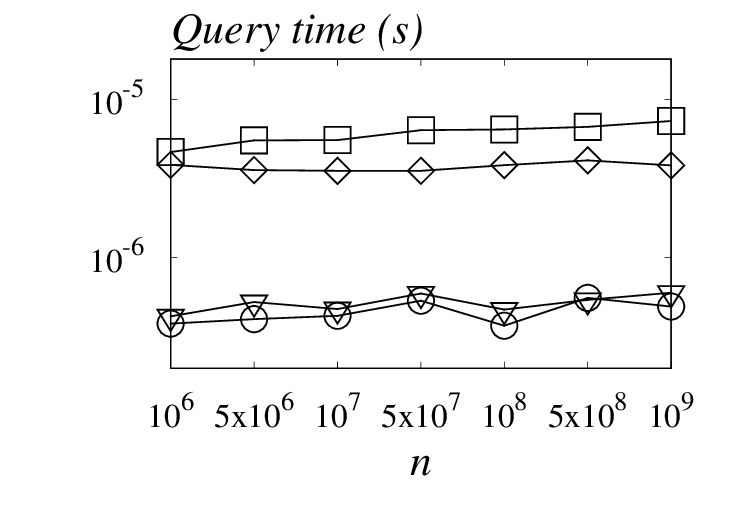}&
    			\hspace{-2mm} \includegraphics[height=26mm]{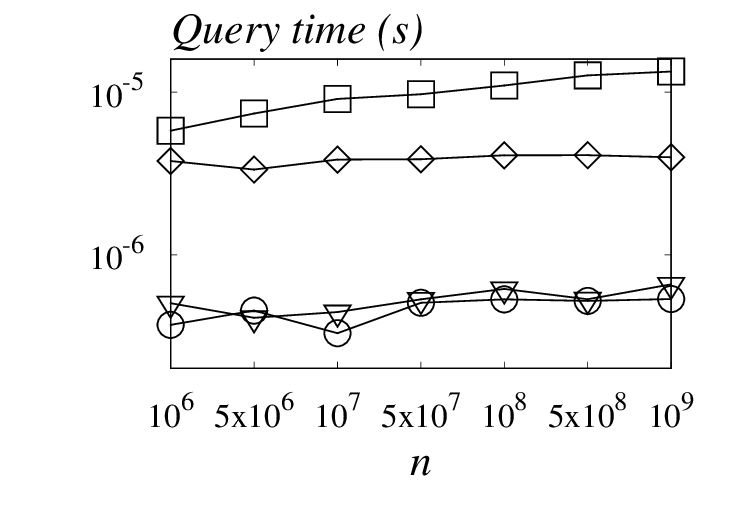}
    			\\[-2mm]
    			\hspace{-1mm} (a) Exponential distribution &
    			\hspace{-1mm} (b) Normal distribution &
    			\hspace{-1mm} (c) Half-normal distribution&
    			\hspace{-1mm} (d) Log-normal distribution \\[-1mm]
    		\end{tabular}
    		\vspace{-2mm}
    		\caption{Varying $n$: Query time (in seconds) on different distributions. (c=0.4)} \label{fig:nvsq4}
    		\vspace{-2mm}
    	\end{small}
    \end{figure*}

\section{Experimental Results for Different $\boldsymbol{c}$ Values}
To evaluate the impact of the sampling parameter $c$, we conduct additional experiments by varying $c$ across three different values: 0.8, 0.6, and 0.4. The query time performance across different probability distributions is presented in Fig.~\ref{fig:nvsq8} (for $c=0.8$), Fig.~\ref{fig:nvsq6} (for $c=0.6$), and Fig.~\ref{fig:nvsq4} (for $c=0.4$). These results complement our main experiments where $c=1.0$ and demonstrate the robustness of our approach across different sampling intensities.

\end{document}